\documentclass[11pt]{article}
\usepackage{amssymb}
\usepackage{amsfonts}
\usepackage{amsmath}
\usepackage{amsthm}
\usepackage{latexsym}

\usepackage{enumitem}

\usepackage{mathtools}

\usepackage{bbm}

\usepackage{xspace} 

\usepackage[pagebackref=true]{hyperref}
\hypersetup{
    unicode=false,          
    colorlinks=true,        
    linkcolor=blue,          
    citecolor=blue,        
    filecolor=magenta,      
    urlcolor=cyan           
}

\usepackage{cleveref}
\usepackage{amsmath}
\usepackage{amsthm}
\usepackage{amsfonts}
\usepackage{fullpage,appendix}

\usepackage[ruled]{algorithm}
\usepackage{algpseudocode}
\usepackage{algorithmicx}

\usepackage{dsfont}
\usepackage{color}
\usepackage{tikz}

\newcommand{\eqdef} {\mathrel{\stackrel{\makebox[0pt]{\mbox{\normalfont\tiny
def}}}{=}}}

\newcommand{\ignore}[1]{}
\definecolor{corlinks}{RGB}{64,128,128}
\definecolor{cormenu}{RGB}{0,37,94}
\definecolor{corurl}{RGB}{0,46,91}
\definecolor{darkgreen}{rgb}{0,0.5,0}

\newcommand{\E} {\mathbb{E}}

\newcommand{\R}{\mathbb{R}}

\newcommand{\N}{\mathbb{N}}

\newcommand{\poly}{\mathsf{poly}}
\newcommand{\Prob}{\Pr}

\newcommand{\eps}{\epsilon}

\newcommand{\ep}{\epsilon}
\newcommand{\ra}{\rightarrow}

\newcommand{\dypv}{{D^{\rm Y}_{\mathrm{PV}}}}
\newcommand{\dnpv}{{D^{\rm N}_{\mathrm{PV}}}}
\newcommand{\dmpv}{{D^{\rm Mix}_{\mathrm{PV}}}}
\newcommand{\Dmpv}{{\mathcal{D}^{\rm Mix}_{\mathrm{PV}}}}
\newcommand{\Dmpvp}{{\mathcal{D}^{\rm Mix+}_{\mathrm{PV}}}}
\newcommand{\One}{\mathbbm{1}}
\newcommand{\indep}{\rotatebox[origin=c]{90}{$\models$}}

\newcommand{\MS}{\mathcal{S}}

\newtheorem{fact}{Fact}[section]
\newtheorem{definition}[fact]{Definition}
\newtheorem{defn}[fact]{Definition}

\newtheorem{theorem}[fact]{Theorem}
\newtheorem{lemma}[fact]{Lemma}
\newtheorem{lem}[fact]{Lemma}
\newtheorem{corollary}[fact]{Corollary}

\newtheorem{proposition}[fact]{Proposition}

\newtheorem{claim}[fact]{Claim}

\newtheorem{remark}[fact]{Remark}

\newcommand{\Inote}[1]{}
\newcommand{\Pnote}[1]{}
\newcommand{\Mnote}[1]{}

\newif\ifnotes\notestrue

\ifnotes

\newcommand{\msnote}[1]{\textcolor{red}{{\bf (Madhu:} {#1}{\bf ) }} \marginpar{\tiny\bf
             \begin{minipage}[t]{0.5in}
               \raggedright Madhu
            \end{minipage}}}
\newcommand{\mbnote}[1]{\textcolor{red}{{\bf (Mitali:} {#1}{\bf ) }} \marginpar{\tiny\bf
        \begin{minipage}[t]{0.5in}
        \raggedright Mitali
    \end{minipage}}}
\newcommand{\bnote}[1]{\textcolor{red}{{\bf (Badih} {#1}{\bf ) }} \marginpar{\tiny\bf
             \begin{minipage}[t]{0.5in}
               \raggedright Badih
            \end{minipage}}}            
\newcommand{\nnote}[1]{\textcolor{red}{{\bf (Noah:} {#1}{\bf ) }} \marginpar{\tiny\bf
             \begin{minipage}[t]{0.5in}
               \raggedright Noah
                \end{minipage}}}  
\else
\newcommand{\msnote}[1]{}
\newcommand{\mbnote}[1]{}
\newcommand{\bnote}[1]{}
\newcommand{\nnote}[1]{}
\fi

\newcommand{\Z}{\mathbb{Z}}

\usepackage{empheq}

\title{Communication-Rounds Tradeoffs for Common Randomness and Secret Key Generation} 

\author{%
Mitali Bafna\thanks{Harvard John A. Paulson School of Engineering and Applied Sciences, 33 Oxford Street, Cambridge, MA 02138, USA. {\tt mitalibafna@g.harvard.edu.} Work supported in part by a Simons Investigator Award and NSF Award CCF 1715187.}
\and
Badih Ghazi\thanks{Google Research, 1600 Amphitheatre Parkway Mountain View, CA 94043, USA. {\tt badihghazi@gmail.com} This work was partly done while the author was a student at MIT. Supported in parts by NSF CCF-1650733 and CCF-1420692.} 
\and 
Noah Golowich\thanks{Harvard University. {\tt ngolowich@college.harvard.edu}}
\and
Madhu Sudan\thanks{Harvard John A. Paulson School of Engineering and Applied Sciences, 33 Oxford Street, Cambridge, MA 02138, USA. {\tt madhu@cs.harvard.edu.} Work supported in part by a Simons Investigator Award and NSF Award CCF 1715187.}
}

\begin{document}
\maketitle

\begin{abstract}
	We study the role of interaction in the Common Randomness Generation (CRG) and Secret Key Generation (SKG) problems. 
	In the CRG problem, two players, Alice and Bob, respectively get samples $X_1,X_2,\dots$ and $Y_1,Y_2,\dots$ with the pairs $(X_1,Y_1)$, $(X_2, Y_2)$, $\dots$ being drawn independently from some known probability distribution $\mu$. They wish to communicate so as to agree on $L$ bits of randomness. The SKG problem is the restriction of the CRG problem to the case where the key is required to be close to random even to an eavesdropper who can listen to their communication (but does not have access to the inputs of Alice and Bob). In this work, we study the relationship between the amount of communication and the number of rounds of interaction in both the CRG and the SKG problems. Specifically, we construct a family of distributions $\mu = \mu_{r, n,L}$, parametrized by integers $r$, $n$ and $L$, such that for every $r$ there exists a constant $b = b(r)$ for which CRG (respectively SKG) is feasible when $(X_i,Y_i) \sim \mu_{r,n,L}$ with $r+1$ rounds of communication, each consisting of $O(\log n)$ bits, but when restricted to $r/2 - 3$
	rounds of interaction, the total communication must exceed $\Omega(n/\log^{b}(n))$ bits. Prior to our work no separations were known for $r \geq 2$.
\end{abstract}

\thispagestyle{empty}

\newpage
\thispagestyle{empty}
\tableofcontents

\newpage
\setcounter{page}{1}

\newpage

\section{Introduction}

\subsection{Problem Definition}

In this work, we study the \emph{Common Randomness Generation (CRG)} and \emph{Secret Key Generation (SKG)} problems  --- two central questions in information theory, distributed computing and cryptography --- and study the need for interaction in solving these problems.
 
In the CRG problem, two players, Alice and Bob, have access to correlated randomness, with Alice being given $X_1, X_2, \dots$, and Bob being given $Y_1, Y_2, \dots$, where $(X_1, Y_1), (X_2, Y_2), \dots$ are drawn i.i.d from some known probability distribution $\mu$. Their goal is to agree on $L$ bits of entropy with high probability while communicating as little as possible. In the SKG problem, the generated random key is in addition required to be secure against a third player, Eve, who does not have access to the inputs of Alice and Bob but who can eavesdrop on their conversation. The CRG and SKG settings are illustrated in Figures~\ref{fig:CRG} and~\ref{fig:SKG} respectively.

Common random keys play a fundamental role in distributed computing and cryptography. They can often be used to obtain significant performance gains that would otherwise be impossible using deterministic or private-coin protocols. Under the additional secrecy constraints, the generated keys are of crucial importance as they can be used for encryption -- a central goal of cryptography.

\begin{figure}[h]
	\hspace{0.65in}
	\includegraphics[scale=0.28]{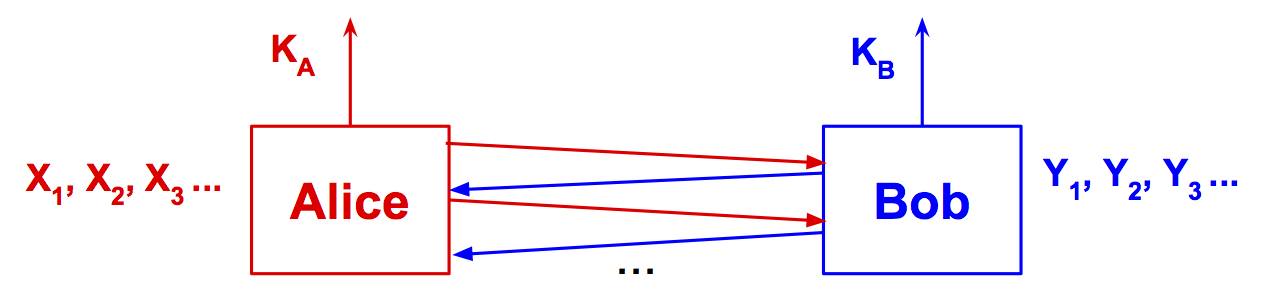}
	\caption{Common Randomness Generation (CRG)}\label{fig:CRG}
\end{figure}

\begin{figure}[h]
	\vspace{0.25in}
	\hspace{0.65in}
	\includegraphics[scale=0.28]{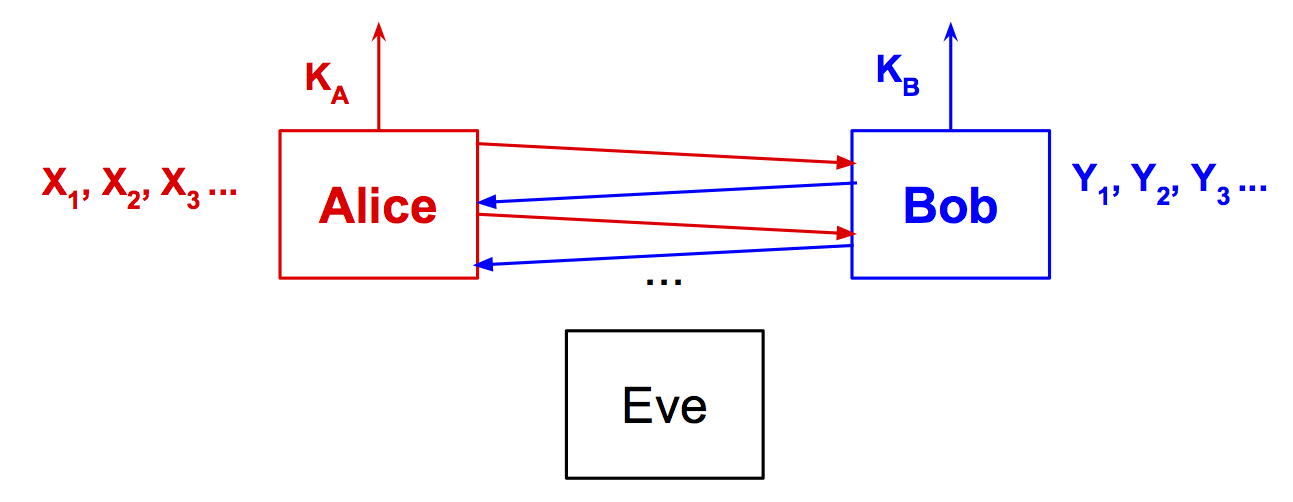}
	\caption{Secret Key Generation (SKG)}\label{fig:SKG}
\end{figure}

This paper investigates the tradeoff between rounds and communication for protocols for common randomness and secret key generation: We start with some terminology needed to describe our problem.
We say that a communication protocol $\Pi$ is an \emph{$(r,c)$-protocol} if it involves at most $r$ rounds of interaction with Alice starting and with the total length of all the messages being at most $c$ bits. Let $H_\infty(\cdot)$ denote the min-entropy function. A protocol is said to be an \emph{$(L, \epsilon)$-CRG} scheme for a correlation source $\mu$ if Alice and Bob get a finite number of i.i.d. samples of $\mu$, and after the final round of $\Pi$, Alice outputs a key $K_A$ and Bob outputs a key $K_B$, with $K_A$ and $K_B$ belonging to a finite set, satisfying $\min\{H_\infty(K_A),H_\infty(K_B)\}  \geq L$, and with $K_A$ and $K_B$ being equal with probability at least $1-\epsilon$. A protocol is said to be an \emph{$(L,\epsilon)$-SKG} scheme for $\mu$ if it is an $(L,\epsilon)$-CRG scheme for $\mu$ and satisfies the additional security guarantee that $\max\{I(\Pi; K_A), I(\Pi; K_B)\} = o(1)$ where $\Pi$ is also used to denote the protocol transcript and $I(\cdot ;  \cdot)$ is the mutual information. Then, we define the $r$-round communication complexity of $(L, \epsilon)$-CRG of a correlation source $\mu$, denoted by $CC_{r}(CRG_{L, \epsilon}(\mu))$, as the smallest $c$ for which there is an $(r,c)$-protocol that is an $(L, \epsilon)$-CRG scheme for $\mu$. We similarly define the $r$-round communication complexity of $(L,\epsilon)$-SKG of $\mu$ and denote it by $CC_r(SKG_{L, \epsilon}(\mu))$.
In terms of the above notation we study the functions $CC_{r}(CRG_{L, \epsilon}(\mu))$ and $CC_r(SKG_{L, \epsilon}(\mu))$ as we vary $r$.

\subsection{History}


The CRG and SKG problems have been well-studied in information theory and theoretical computer science. In information theory, they go back to the seminal work of Shannon on secrecy systems \cite{shannon1949communication}, which was followed by the central works of  Maurer \cite{maurer1993secret} and Ahlswede and Csisz{\'a}r \cite{ahlswede1993common, ahlswede1998common}. A crucial motivation for the study of SKG is the task of \emph{secure encryption}, where a common secret key can potentially be used to encrypt/decrypt messages over an insecure channel. It turns out that without \emph{correlated} inputs (and even allowing each party an unlimited amount of private randomness), efficiently generating common randomness is infeasible: agreeing on $L$ bits of randomness with probability $\gamma$ can be shown to require communicating at least $L-O(\log(1/\gamma))$ bits~\footnote{This fact is a special case of several known results in the literature on CRG. In particular, it follows from the proof of the agreement distillation lower bound of \cite{canonne2017communication}}. Since the original work of Shannon, the questions of how much randomness can be agreed on, with what probability, with what type of correlation and with how many rounds of interaction have attracted significant effort in both the information theory and theoretical computer science communities (e.g., \cite{maurer1993secret,ahlswede1993common, ahlswede1998common,csiszar2000common,gacs1973common,Wyner_CommonInfo, csiszar2004secrecy,zhao2011efficiency,tyagi2013common, liu2015secret,liu2016common, bogdanov2011extracting,  chan2014extracting, guruswami2016tight, ghazi2018resource} to name a few). In particular, Ahlswede and Csisz{\'a}r studied the CRG and SKG problems in the case of \emph{one-way communication} where they gave a characterization of the ratio of the entropy of the key to the communication in terms of the \emph{strong data processing constant} of the source (which is closely related to its hypercontractive properties \cite{ahlswede1976spreading, anantharam2013maximal}).

We point out that the aforementioned results obtained in the information theory community hold for the \emph{amortized} setup where the aim is to characterize the achievable $(H,C)$ pairs for which for every positive $\delta$, there is a large enough $N$, such that there is a CRG/SKG scheme taking as input $N$ i.i.d. copies from the source and generating $(H-\delta)\cdot N$ bits of entropy while communicating at most $(C+\delta) \cdot N$ bits. Moreover, these results mostly focus on the regime where the agreement probability gets arbitrarily close to one for sufficiently large $N$. The \emph{non-amortized} setup, where the entropy of the keys and the communication are potentially independent of the number of i.i.d. samples drawn from the source, as well as the setting where the agreement probability is not necessarily close to one, have been studied in several works within theoretical computer science. In particular, for the \emph{doubly symmetric binary source}, Bogdanov and Mossel gave a CRG protocol with a nearly tight agreement probability in the \emph{zero-communication} case where Alice and Bob are not allowed to communicate \cite{bogdanov2011extracting}. This CRG setup can be viewed as an abstraction of practical scenarios where hardware-based procedures are used for extracting 
a unique random ID from process 
variations~\cite{lim2005extracting,su2008digital,yu2009towards}
that can then be used for authentication~\cite{lim2005extracting,suh2007physical}. Guruswami and Radhakrishnan generalized the study of Bogdanov and Mossel to the case of one-way communication (in the non-amortized setup) where they gave a protocol achieving a near-optimal tradeoff between (one-way) communication and agreement probability \cite{guruswami2016tight}. Later, \cite{ghazi2018resource} gave explicit and sample-efficient CRG (and SKG) schemes matching the bounds of \cite{bogdanov2011extracting} and \cite{guruswami2016tight} for the doubly symmetric binary source and the bivariate Gaussian source.

Common randomness is thus a natural model for studying how shared keys can be generated
in settings where only weaker forms of correlation are available. It is one of the simplest and most natural questions within the study of \emph{correlation distillation} and the \emph{simulation of joint distributions} ~\cite{gacs1973common, Wyner_CommonInfo, witsenhausen1975sequences, mossel2004coin,mossel2006non,kamath2015non, ghazi2016decidability, de2018non, ghazi2017dimension}.

Moreover, when studying the setup of \emph{communication with imperfectly shared randomness}, Canonne et al. used lower bounds for CRG \emph{as a black box} when proving the existence of functions having small communication complexity with public randomness but large communication complexity with imperfectly shared randomness \cite{canonne2017communication}. Their setup -- which interpolates between the extensively studied public-coin and private-coin models of communication complexity -- was first also independently introduced by \cite{bavarian2014role} and further studied in \cite{ghazi2015communication,ghazi2018resource}.

Despite substantial work having been done on CRG and SKG, some very basic questions remained open such as the the quest of this paper, namely the role of interaction in generating common randomness (or secret keys). 
Recently, Liu, Cuff and Verdu generalized the CRG and SKG characterizations of Ahlswede and Csisz{\'a}r to the case of \emph{multi-round communication} \cite{liu2015secret,liu2016common,liurate}. Their characterization has been shown by \cite{ghazi2018resource} to be intimately connected to the notions of \emph{internal} and \emph{external information costs} of protocols which were first defined by \cite{BJKS,barak2013compress} and \cite{cswy01} respectively (who were motivated by the study of direct-sum questions arising in theoretical computer science). However their work does not yield sources for which randomness generation requires many rounds of interaction (to be achieved with low commununication). Their work does reveal sources where interaction does {\em not} help. For example, in the case where the agreement probability tends to one, Tyagi had shown that for \emph{binary symmetric} sources, interaction does not help, and conjectured the same to be true for any (possibly asymmetric) \emph{binary} source \cite{tyagi2013common}-- a conjecture which was proved by Liu, Cuff and Verdu \cite{liu2016common}. Morever, Tyagi constructed a source on \emph{ternary alphabets} for which there is a constant factor gap between the $1$-round and $2$-round communication complexity for Common Randomness and Secret Key Generation.  This seems to be the strongest tradeoff known for communication complexity of CRG or SKG till our work. 




\subsection{Our Results}

In this work, we study the relationship between the \emph{amount of communication} and the \emph{number of rounds} of interaction in each of the CRG and SKG setups, namely: can Alice and Bob communicate less and still generate a random/secret key by interacting for a larger number rounds?

For every constant $r$ and parameters $n$ and $L$, we construct a family of probability distributions $\mu = \mu_{r, n, L}$ for which CRG (respectively SKG) is possible with $r$ rounds of communication, each consisting of $O(\log{n})$ bits, but when restricted to $r/2$ rounds, the total communication of any protocol should exceed $n/\log^{\omega(1)}(n)$ bits. Formally, we show that $CC_{r+1}(CRG_{L, 0}(\mu)) \leq (r+1)\log n$ while for every constant $\epsilon < 1$ we have that $CC_{r/2-3}(CRG_{\ell,\epsilon}) \geq \min\{\Omega(\ell),n/\poly\log n\}$ (and similarly for SKG). 

\begin{theorem}[Communication-Rounds Tradeoff for Common Randomness Generation]
	\label{thm:main-crg}
	For all $\epsilon < 1, r \in \Z^+$, there exist $\eta > 0, n_0, \beta < \infty$, such that for all $n\geq n_0, L$ there exists a source $\mu_{r,n,L}$ for which the following hold:
	\begin{enumerate}
		\item There exists an $((r+1),(r+1)\lceil \log n\rceil)$-protocol for $(L,0)$-CRG from $\mu_{r,n,L}$.
		\item For every $\ell \in \Z^+$ there is no $(r/2-2,\min\{\eta \ell - \beta, n/\log^\beta n\})$-protocol for $(\ell,\epsilon)$-CRG from $\mu_{r,n,L}$.
	\end{enumerate}
\end{theorem}

We also get an analogous theorem for SKG, with the same source!

\begin{theorem}[Communication-Rounds Tradeoff for Secret Key Generation]
	\label{thm:main-skg}
	For all $\epsilon < 1, r \in \Z^+$, there exist $\eta > 0, n_0, \beta < \infty$, such that for all $n\geq n_0, L$ there exists a source $\mu_{r,n,L}$ for which the following hold:
	\begin{enumerate}
		\item There exists an $((r+1),(r+1)\lceil \log n\rceil)$-protocol for $(L,0)$-SKG from $\mu_{r,n,L}$.
		\item For every $\ell \in \Z^+$ there is no $(r/2-2,\min\{\eta \ell - \beta, n/\log^\beta n\})$-protocol for $(\ell,\epsilon)$-SKG from $\mu_{r,n,L}$.
	\end{enumerate}
\end{theorem}

In particular, our theorems yield a gap in the amount of communication that is almost \emph{exponentially large} if the number of rounds of communication is squeezed by a constant factor. Note that every communication protocol can be converted to a two-round communication protocol with an exponential blowup in communication - so in this sense our bound is close to optimal.
Prior to our work, no separations were known for any number of rounds larger than two!

\subsection{Brief Overview of Construction and Proofs}

Our starting point for constructing the source $\mu$ is the well-known ``pointer-chasing'' problem \cite{NW} 
used to study tradeoffs between rounds of interaction and communication complexity.
In (our variant of) this problem Alice and Bob get a series of permutations $\pi_1,\pi_2,\ldots,\pi_r: [n]\to[n]$ along with an initial pointer $i_0$ and their goal is to ``chase'' the pointers, i.e., compute $i_r$ where $i_j = \pi_j(i_{j-1})$ for every $j \in \{1,\dots, r\}$.
Alice's input consists of the odd permutations $\pi_1,\pi_3,\ldots,$ and Bob gets the initial pointer $i_0$ and the even permutations $\pi_2,\pi_4,\ldots$. The natural protocol to determine $i_r$ takes $r+1$ rounds of communication with the $j$th round involving the message $i_j$ (for $j = 0,\ldots,r$). Nisan and Wigderson show that any protocol with $r$ rounds of interaction requires $\Omega(n)$ bits of communication~\cite{NW}.

To convert the pointer chasing instance into a correlated source, we let the source include $2n$ strings $A_1,\ldots,A_n$ and $B_1,\ldots,B_n \in \{0,1\}^L$ where $(A_1,\ldots,B_n)$ is uniform in $\{0,1\}^{2nL}$ conditioned on $A_{i_r} = B_{i_r}$. Thus the source outputs $X = (\pi_1,\pi_3,\ldots; A_1,\ldots,A_n)$ and
$Y = (i_0,\pi_2,\pi_4,\ldots; B_1,\ldots,B_n)$ satisfy $A_{i_r} = B_{i_r}$ with
$i_j = \pi_j(i_{j-1})$ for every $j \in \{1, \dots, r\}$. (See \Cref{def:PCS} and \Cref{fig:PCS} for more details.) The natural protocol for the pointer chasing problem also turns into a natural protocol for CRG and SKG with $r+1$ rounds of communication, and our challenge is to show that protocols with few rounds cannot extract randomness.

The lower bound does not follow immediately from the lower bound for the pointer chasing problem --- and indeed we do not even give a lower bound for $r - O(1)$ rounds of communication. We explain some of the challenges here and how we overcome them. 

Our first challenge is that there is a low-complexity ``non-deterministic protocol'' for common randomness generation in our setting. The players somehow guess $i_r$ and then verify $A_{i_r} = B_{i_r}$ (by exchanging the first $\log1/\epsilon$ bits of these strings) and if they do, then they output $A_{i_r}$ and $B_{i_r}$ respectively. While the existence of a non-deterministic protocol does not imply the existence of a deterministic one, it certainly poses hurdles to the lower bound proofs. Typical separations between non-deterministic communication complexity and deterministic ones involve lower bounds such as those for ``set-disjointness''~\cite{KalyanaS,Razborov,BJKS} which involve different reasoning than the ``round-elimination'' arguments in \cite{NW}. 
Our lower bound would somehow need to combine the two approaches.

We manage to do so ``modularly'' at the expense of a factor of $2$ in the number of rounds of communication by introducing an intermediate ``pointer verification (PV)'' problem. In this problem Alice and Bob get permutations $\pi_1,\ldots,\pi_r$ (with Alice getting the odd ones and Bob the even ones) and additionally Bob gets pointers $i$ and $j$. Their goal is to decide if the final pointer $i_r$ equals $j$ given that the initial pointer $i_0$ is equal to $i$. 
The usefulness of this problem comes from the fact that we can reduce the common randomness generation problem to the complexity of the pointer verification problem on a specific (and natural) distribution: Specifically if PV is hard on
this distribution with $r'$ rounds of communication, then we can show (using the hardness of set disjointness as a black box) that 
the common randomness generation problem is hard with $r'-1$ rounds of communication. 

We thus turn to showing lower bounds for PV. We first note that we cannot expect a lower bound for $r$ rounds of communication: 
PV can obviously be solved in $r/2$ rounds of communication with Alice and Bob chasing both the initial and final pointers till they meet in the middle. We also note that one can use the lower bound from \cite{NW} as a black box to get a lower bound of $r/2 - 1$ rounds of communication for PV but it is no longer on the ``natural'' distribution we care about and thus this is not useful for our setting.

The bulk of this paper is thus devoted to proving an $r/2 - O(1)$ round lower bound for the PV problem on our distribution. We get this lower bound by roughly following the ``round elimination'' strategy of \cite{NW}. A significant challenge in extending these lower bounds to our case is that we have to deal with distributions where Alice and Bob's inputs are dependent. This should not be surprising since the CRG problem provides Alice and Bob with correlated inputs, and so there is resulting dependency between Alice and Bob even before any messages are sent. The dependency gets more complex as Alice and Bob exchange messages, and  
we need to ensure that the resulting mutual information is not correlated with the desired output, i.e., the PV value of the game. We do so by a delicate collection of conditions (see \Cref{def:dmixdef}) that allow the inputs to be correlated while guaranteeing sufficient independence to carry out a round elimination proof. See \Cref{sec:proofs} for details.


\paragraph{Organization of Rest of the Paper.}
In \Cref{sec:constr}, we present our construction of the distribution $\mu$ alluded to in \Cref{thm:main-crg} and \Cref{thm:main-skg}. In \Cref{sec:reduction} we reduce the task of proving communication lower bounds for CRG with few rounds to the task of proving lower bounds for distinguishing some distributions. We then introduce our final problem, the Pointer Verification problem, and the distribution on which we need to analyze it in \Cref{sec:PV_prob}. This section includes the statement of our main technical theorem about the pointer verification problem (\Cref{thm:pv-main}) and the proofs of \Cref{thm:main-crg} and \Cref{thm:main-skg} assuming this theorem.
Finally in \Cref{sec:proofs}, we prove \Cref{thm:pv-main}.

\section{Construction}
\label{sec:constr}

We start with some basic notation used in the rest of the paper.
For any positive integer $n$, we denote by $[n]$ the set $\{1, \dots, n\}$. We use $\log$ to denote the logarithm to the base $2$. For a distribution $D$ on a universe $\Omega$ we use the notation $X \sim D$ to denote a random variable $X$ sampled according to $D$. For any positive integer $t$, we denote by $D^t$ the distribution obtained by sampling $t$ independent identically distributed samples from $D$. We use the notation $X \indep Y$ to denote that $X$ is independent of $Y$ and $X\indep Y | Z$ to denote that $X$ and $Y$ are independent conditioned on $Z$. 
We denote by $\E_{X \sim D}[X]$ the expectation of $X$ and for an event $E \subseteq \Omega$, we denote by $\Pr_X[E]$  the probability of the event $E$. 
For $i \in \Omega$, $D_i$ (and sometimes $D(i)$) denotes the probability of the element $i$, i.e., $D_i = D(i) = \Pr_{X \sim D} [X = i]$. 
For distributions $P$ and $Q$ on $\Omega$, the total variation distance $\Delta(P,Q) \eqdef \frac 12 \sum_{i \in \Omega} |P_i - Q_i|$. The entropy of $X \sim P$ is the quantity $H(X) = \E_{X \sim P}[-\log P_X]$. The min-entropy of $X \sim P$ is the quantity $H_\infty(X) = \min_{x \in \Omega}\{-\log P_x\}$. For a pair of random variables $(X,Y) \sim P$, $P_X$ denotes the marginal distribution on $X$ and $P_{X|y}$ denotes the distribution of $X$ conditioned on $Y = y$. The conditional entropy $H(X|Y) \eqdef \E_{y \sim P_Y} [H(X_y)]$, where $X_y \sim P_{X|Y=y}$. The mutual information between $X$ and $Y$, denoted $I(X;Y)$, is the quantity $H(X) - H(X|Y)$. The conditional mutual information between $X$ and $Y$ conditioned on $Z$, denoted $I(X;Y | Z)$, is the quantity $\E_{z \sim P_Z}[H(X_z) - H(X_z | Y_z)]$ where $(X_z,Y_z) \sim P_{X,Y | Z= z}$.  
We use standard properties of entropy and information such as the Chain rules and the fact ``conditioning does not increase entropy''.
For further background material on information theory and communication complexity, we refer the reader to the books \cite{cover2012elements} and \cite{kushilevitz_nisan} respectively.

We start by describing the family of distributions $\mu_{r,n,L}$ that we use to prove \Cref{thm:main-crg} and \Cref{thm:main-skg}. For a positive integer $n$, we let $S_n$ denote the family of all permutations of $[n]$. 

\begin{defn}[The Pointer Chasing Source $\mu_{r,n,L}$]
\label{def:PCS}
	For positive integers $r$, $n$ and $L$, the support of $\mu = \mu_{r,n,L}$ is 
	$(S_n^{\lceil r/2\rceil} \times \{0,1\}^{nL}) \times ([n]\times S_n^{\lfloor r/2 \rfloor}\times \{0,1\}^{nL})$. Denoting $X = (\pi_1,\pi_3,\ldots,\pi_{2\lceil r/2 \rceil-1},A_1,\ldots,A_n)$ and $Y = (i,\pi_2,\pi_4,\ldots,\pi_{2\lfloor r/2 \rfloor},B_1,\ldots,B_n)$, a sample $(X,Y) \sim \mu$ is drawn as follows:
	\begin{itemize}
		\item $i \in [n]$ and $\pi_1,\ldots,\pi_r \in S_n$ are sampled uniformly and independently. 
		\item Let $j = \pi_r(\pi_{r-1}(\cdots \pi_1(i)\cdots))$. 
		\item $A_j = B_j \in \{0,1\}^L$ is sampled uniformly and independently of $i$ and $\pi$'s. 
	    \item For every $k \ne j$, $A_k\in \{0,1\}^L$ and $B_k\in \{0,1\}^L$ are sampled uniformly and independently.
	\end{itemize}
    See \Cref{fig:PCS} for an illustration of the inputs to the Pointer Chasing Source.
\end{defn}

\begin{figure}[h]
	\hspace{1in}
	\includegraphics[scale=0.5]{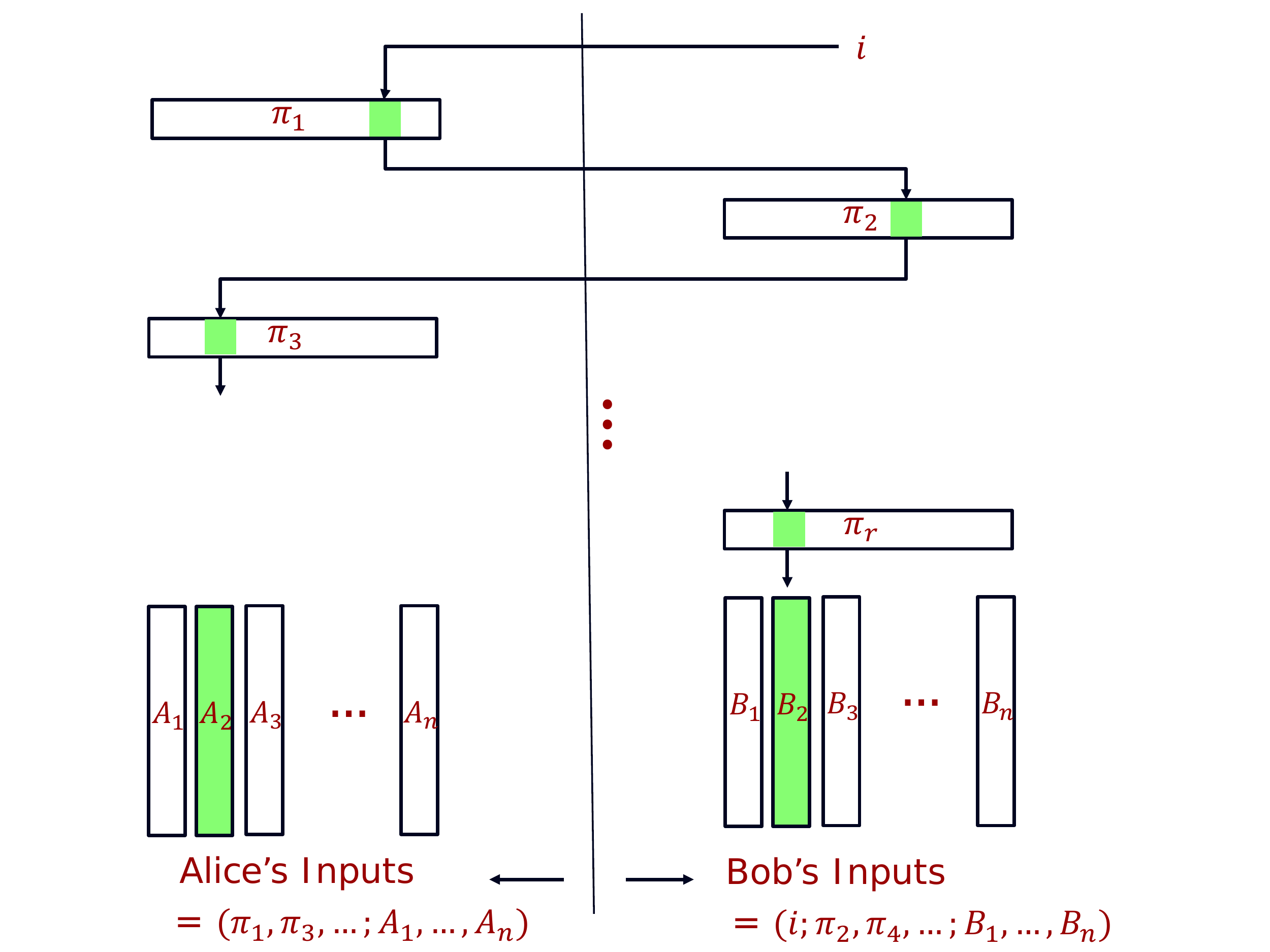}
	\caption{The Pointer Chasing Source}\label{fig:PCS}
\end{figure}

Informally, a sample from $\mu$ contains a common hidden block of randomness $A_j = B_j \in \{0,1\}^L$ that Alice and Bob can find by following a sequence of pointers, where Alice holds the odd pointers in the sequence and Bob holds the even pointers. The next lemma gives (the obvious) upper bound on the $r$-round communication needed to generate common randomness from $\mu$.

\begin{lem}[Upper bound on $r$-round communication of SKG]
	\label{lem:ub-skg}
	For every $r$, $n$ and $L$, there exists an $(r+1,\lceil \log n \rceil)$-protocol for $(L,0)$-SKG (and hence also for $(L,0)$-CRG) from $\mu_{r,n,L}$ with Bob speaking in the first round. 
\end{lem}

\begin{proof}
	The protocol $\Pi$ is the obvious one in which Bob and Alice alternate by sending a pointer to each other starting with $i$ and culminating in $j$, and the randomness they ``agree on'' is $A_j = B_j$. 
	
	Formally, for $t \in [r]$, let $i_{t} = \pi_t(i_{t-1})$ with $i_0 = i$. In odd round $t+1$, Bob sends $i_t$ to Alice and in even round $t+1$, Alice sends $i_t$ to Bob. At the end of $r+1$ rounds of communication Alice outputs $A_{i_{r}}$ and Bob outputs $B_{i_{r}}$.
	
	Note that by the construction of $\mu$, we have that $i_{r} = j$ and $A_j = B_j$. Note further that at the beginning of the $(t+1)$st round of communication both  Alice and Bob know $i_{t-1}$. Furthermore if $t+1$ is odd, then Bob also knows $\pi_{t}$ and hence can compute $i_t = \pi_{t}(i_{t-1})$ (and similarly Alice knows her message in even rounds). 
	
	Thus we conclude that the above is a valid $(r+1,\lceil \log n\rceil)$-protocol for $(L,0)$-CRG. Furthermore since $A_{i_r} = B_{i_r}$ is independent of $i_0,\ldots,i_r$ it follows that $I((i_0,\ldots,i_r); A_{i_r}) = I(\Pi;K_A) = 0$ (and similarly for $I(\Pi;K_B)$) and so this is also a valid protocol $(L,0)$-SKG.
\end{proof}

In the rest of the paper we show that no $r/2 - O(1)$ round protocol can solve CRG from $\mu_{r,n,L}$ with non-trivial communication.

\section{Related Indistinguishability Problems}\label{sec:reduction}

Our lower bound on the number of rounds needed to generate common randomness comes from an ``indistinguishability argument''. We show that to protocols with a small number of rounds and small amount of communication, the distribution $\mu$ is indistinguishable from the distribution $\mu_X \times \mu_Y$, where Alice and Bob's inputs are independent. Using the well-known fact that generating $L$ bits of common randomness essentially requires $L$ bits of communication in the absence of correlated inputs, this leads us to conclude that CRG is hard with limited number of rounds of communication. 

In this section we simply set up the stage by defining the notion of indistinguishability and connecting it to the task of common randomness generation, leaving the task of proving the indistinguishability to later sections.

\subsection{The Main Distributions and Indistinguishability Claims}

We start by defining the indistinguishability of inputs to protocols.

\begin{definition}
We say that two distributions $D_1$ and $D_2$ on $(X,Y)$ are {\em $\epsilon$-indistinguishable} to a protocol $\Pi$ if the distributions of transcripts (the sequence of messages exchanged by Alice and Bob) generated when $(X,Y) \sim D_1$ has total variation distance at most $\epsilon$ from the distribution of transcripts when $(X,Y) \sim D_2$. 

We say that distributions $D_1$ and $D_2$ are {\em $(\epsilon,c,r)$-indistinguishable} if they are $\epsilon$-indistinguishable to every $(r,c)$-protocol $\Pi$ using public randomness. Conversely, we say that the distributions $D_1$ and $D_2$ are {\em $(\ep, c, r)$-distinguishable} if they are not $(\ep, c, r)$-indistinguishable.
\end{definition}

Fix $r,n,L$ and let $\mu = \mu_{r,n,L}$. Now let $\mu_X$ denote the marginal distribution of $X$ under $\mu$, i.e., $X = (\pi_1,\pi_3,\ldots,\pi_{2\lceil r/2 \rceil-1},A_1,\ldots,A_n)$ have all coordinates chosen independently and uniformly from their domains. Similarly let $\mu_Y$ denote the marginal on $Y$, and let $\mu_X \times \mu_Y$ denote the distribution where $X \sim \mu_X$ and $Y \sim \mu_Y$ are chosen independently. 

Our main technical result (\Cref{thm:pv-main} and in particular its implication \Cref{lem:main-ind}) shows that $\mu$ and $\mu_X \times \mu_Y$ are $(\epsilon,r/2-O(1),n/\poly \log n)$-indistinguishable, even to protocols with common randomness. In the rest of this section, we explain why this rules out common randomness generation.

\subsection{Reduction to Common Randomness Generation}

\begin{proposition}
\label{prop:priv-impos}
	There exists a constant $\eta > 0$ such that for every $r,r',n,L,\ell,t$ and $\eps<1$, there is no $(r',\eta \ell-\log(1/1-\eps))$-protocol for $(\ell,\epsilon)$-CRG from $\mu_X^t \times \mu_Y^t$, where $\mu = \mu_{r,n,L}$ with $\mu_X$ and $\mu_Y$ being its marginals.
\end{proposition}

\begin{proof}
	This is essentially folklore. For instance it follows immediately from  \cite[Theorem 2.6]{canonne2017communication} using $\rho=0$ (which corresponds to private-coin protocols). 
\end{proof}

\begin{proposition}	
  \label{prop:need-dist}
  There is an absolute constant $\xi$ such that the following holds.
	Let $\eta$ be the constant from \Cref{prop:priv-impos}.
	If there exists an $(r',c)$-protocol that solves the $(\ell,1 - \gamma)$-CRG problem from $\mu = \mu_{r,n,L}$ with $c < \eta(\ell-3) - \log 1/\gamma$, then there exists some positive integer $t$ for which $\mu^t$
	and $\mu_X^t \times \mu_Y^t$ are  $(\gamma/10,r'+1,c+ \xi \log1/\gamma)$-distinguishable.
\end{proposition}

\begin{proof}
	Let $\Pi$ be an $(r',c)$ protocol with private randomness for $(\ell,1-\gamma)$-CRG from $\mu$ and let $D_1$ denote the distribution of $K_A$ conditioned on $K_A = K_B$. Let $t$ be the number of samples of $\mu$ used by $\Pi$. Let $I = \One[K_A = K_B]$ be the indicator variable determining if $K_A = K_B$. Let $D^A_1$ be the distribution of $(K_A,I)$ when $\Pi$ is run on samples from $\mu^t$. Let $D^A_2$ be the distribution of the $(K_A,I)$ when $\Pi$ is run on samples from $\mu_X^t \times \mu_Y^t$. Define $D^B_1$ and $D^B_2$ analogously. We distinguish between the cases where $\Delta(D^A_1,D^A_2)$
	and $\Delta(D^B_1,D^B_2)$ are both small from the cases where one of them is large.

	\medskip 
	
	\noindent {\sf Case 1: $D^A_1$ is $\gamma/4$-far from $D^A_2$ (in total variation distance).} We argue that in this case, $\mu^t$ and $\mu_X^t \times \mu_Y^t$ are distinguishable. Let $T$ be the optimal distinguisher of $D^A_1$ from $D^A_2$ (i.e., $T$ is a $0/1$ valued function with $\E_{(K_A,I) \sim D^A_1}[T(K_A,I)] - \E_{(K_A,I) \sim D^A_2}[T(K_A,I)] \geq \gamma/4$).
	Let $\alpha$ denote $\E_{(K_A,I) \sim D^A_2}[T(K_A,I)]$. 
	We now describe a protocol $\Pi'$ which uses public randomness and augments $\Pi$ by including a bit $I'$ (which is usually equal to $I$) and $T(K_A,I')$ as part of the transcript. We consider two subcases: (1) If Bob is the last speaker in $\Pi$, then $\Pi'$ executes $\Pi$ and then at the conclusion of $\Pi$, Bob sends a random hash $h_B = h(K_B)$ which is $O(\log 1/\gamma)$ bits long (so that for $K_A \ne K_B$ we have $\Pr_h [h(K_A) = h(K_B)]\leq \gamma/20$). Alice then sends $I' = \One[h(K_A) = h_B]$ and the bit $b_{I'}  = T(K_A,I')$. (2) If Alice is the last speaker in $\Pi$, then $\Pi'$ executes $\Pi$ and then Alice sends $h_A = h(K_A)$ to Bob, as well as $b_0 = T(K_A,0)$ and $b_1 = T(K_A,1)$. Bob then sends $I' = \One[h_A = h(K_B)]$ and $b_{I'}$. 
	
	Note that in both cases $\Pi'$ has $r'+1$ rounds of communication and the total number of bits of communucation is $c + O(\log 1/\gamma)$. We now show that $\Pi'$ distinguishes $\mu^t$ from  $\mu_X^t \times \mu_Y^t$ with probability $\Omega(\gamma)$. To see this note that $\Pr_{(K_A,I) \sim D^A_1}[b_{I'} = 1] \geq \Pr_{(K_A,I) \sim D^A_1}[T(K_A,I) = 1] - \Pr_h[I' \ne \One[K_A = K_B]] \geq (\alpha + \gamma/4) - \gamma/20 = \alpha + \gamma/5$. 
	On the other hand we also have $\Pr_{(K_A,I) \sim D^A_2}[b_{I'} = 1] \leq \Pr_{(K_A,I) \sim D^A_2}[T(K_A,I) = 1] + \Pr_h[I' \ne \One[K_A = K_B]] \leq \alpha + \gamma/20$. We conclude
	that $\Pr_{(K_A,I) \sim D^A_1}[b_{I'} = 1] - \Pr_{(K_A,I) \sim D^A_2}[b_{I'} = 1] \geq \gamma/5 - \gamma/20 \geq \gamma/10$. And since $b_{I'}$ is a part of the transcript of $\Pi'$ we conclude that the two distributions are $\gamma/10$-distinguished by $\Pi'$.
	
	\medskip 
	
	\noindent {\sf Case 2: $D^B_1$ is $\gamma/4$-far from $D^B_2$.} This is similar to the above and yields that $\mu^t$
	and $\mu_X^t \times \mu_Y^t$ are $(\gamma/10,r'+1,c+ O(\log1/\gamma))$-distinguishable.
	
		\medskip 
	
	\noindent {\sf Case 3: $\Delta(D^A_1,D^A_2) \leq \gamma/4$ and $\Delta(D^B_1,D^B_2) \leq \gamma/4$.} We argue that this case can not happen since this allows a low-communication protocol to solve CRG with private randomness, thereby contradicting \Cref{prop:priv-impos}. The details are the following.
	
	Our main idea here is to run $\Pi$ on $\mu_X^t \times \mu_Y^t$ (which, being a product distribution involves only private randomness). The proximity of $D^A_1$ to $D^A_2$ implies that the probability that $K_A = K_B$ when $\Pi$ is run on $\mu_X^t \times \mu_Y^t$ is at least $3\gamma/4$ (since the probability that $K_A = K_B$ on $\mu^t$ is at least $\gamma$ and the probability that $I = \One[K_A = K_B]$ is different under $\mu^t$ than under $\mu_X^t \times \mu_Y^t$ is at most $\gamma/4$). But we are not done since the min-entropy of $K_A$ or $K_B$ when $\Pi$ is run on $\mu_X^t \times \mu_Y^t$ might not be lower-bounded by $\ell$. So we modify $\Pi$ to get a protocol $\Pi'$ as follows: Run $\Pi$ and let $(K_A,K_B)$ be the output of $\Pi$. (The output of $\Pi'$ will be different as we see next.)  If the probability of 
	outputting $K_A$ is more than $4 \cdot 2^{-\ell}$ then let $K'_A$ be a uniformly random string in $\{0,1\}^\ell$, else let $K'_A = K_A$. Similarly if the probability of 
	outputting $K_B$ is more than $4 \cdot 2^{-\ell}$ then let $K'_B$ be a uniformly random string in $\{0,1\}^\ell$, else let $K'_B = K_B$. (Note that when $K'_A \ne K_A$ then $K'_A$ and $K'_B$ are independent.) Let $(K'_A,K'_B)$ be the outputs of $\Pi'$. We claim below that $\Pi'$ solves the $(\ell-3,1-\gamma/12)$-CRG from $\mu_X^t \times \mu_Y^t$ which contradicts \Cref{prop:priv-impos} if $c < \eta(\ell-3) - \log(12/\gamma)$. First note that by design the probability of outputting any fixed output $k'_A$ is at most $4\cdot 2^{-\ell} + 2^{-\ell} < 2^{-(\ell-3)}$. (If $\Pr[K_A = k'_A] \geq 4\cdot 2^{-\ell}$ then $\Pr[K'_A = k'_A] \leq 2^{-\ell}$, else $\Pr[K'_A = k'_A] \leq \Pr[K_A = k'_A] + 2^{-\ell}$.) It remains to see that $\Pr[K'_A = K'_B] \geq \gamma/12$. First note that $\Pr[K_A \ne K'_A] \leq \gamma/3$. This is so since every $k'_A$ such that $\Pr[K_A = k'_A] \geq 4 \cdot 2^{-\ell}$ contributes at least $\Pr[K_A = k'_A] - 2^{-\ell} \geq (3/4)\cdot \Pr[K_A = k'_A]$ to $\Delta(D^A_1,D^A_2)$ (the probability of $k'_A$ on $\mu^t$ is at most $2^{-\ell}$). Thus using $\Delta(D^A_1,D^A_2) \leq \gamma/4$, we conclude $\Pr[K_A \ne K'_A] \leq (4/3)\Delta(D^A_1,D^A_2) \leq \gamma/3$. But now we have $\Pr[K'_A = K'_B] \geq \Pr[K_A = K_B] - (\Pr[K_A \ne K'_A]+\Pr[K_B \ne K'_B]) \geq 3\gamma/4 - 2\gamma/3 = \gamma/12$.

\end{proof}

\subsection{Reduction to the Case $t=1$}

Next we show that we can work with the case $t=1$ without loss of generality. Roughly the intuition is that all permutations look the same, and so chasing one series of pointers $\pi_1,\ldots,\pi_r$ is not harder than chasing a sequence of $t$ pointers of the form $(\pi'_{1,\tau}\ldots,\pi'_{r,\tau})_{\tau\in [t]}$. Informally, even if the players in latter problem are given the extra information $(\pi'_{\ell,\tau})^{-1}\pi_\ell$, for every $\ell \in [r]$ and $\tau \in [t]$, they still have to effectively chase the pointers $\pi_1,\ldots,\pi_r$. This intuition is formalized in the reduction below.

\begin{proposition}
	\label{prop:t-removal}
	Fix $r,n,L$ and let $\mu = \mu_{r,n,L}$ and $\mu_X$ and $\mu_Y$ be its marginals. If there exists $\epsilon,r',c,t$ such that $\mu^t$ and $\mu_X^t \times \mu_Y^t$ are $(\epsilon,r',c)$-distinguishable, then $\mu'= \mu_{r,n,Lt}$ and $(\mu')_X \times (\mu')_Y$ are $(\epsilon,r',c)$-distinguishable.
\end{proposition}

\begin{proof}
	Suppose $\Pi$ is a $(r',c)$-protocol that $\epsilon$-distinguishes $\mu^t$ from $\mu_X^t \times \mu_Y^t$. We show how to distinguish $\mu'$ from $(\mu')_X \times (\mu')_Y$ using $\Pi$. Let $(X,Y)$ be an instance of the $\mu'$ vs. $(\mu')_X \times (\mu')_Y$ distinguishability problem.  We now show how Alice and Bob can use common randomness to generate $(X'_1,Y'_1),\ldots,(X'_t,Y'_t)$ such that $((X'_1,Y'_1),\ldots,(X'_t,Y'_t)) \sim \mu^t$ if $(X,Y) \sim \mu'$ and $((X'_1,Y'_1),\ldots,(X'_t,Y'_t)) \sim \mu_X^t\times \mu_Y^t$ if $(X,Y) \sim \mu'_X \times \mu'_Y$. It follows that by applying $\Pi$ to $((X'_1,Y'_1),\ldots,(X'_t,Y'_t))$, Alice and Bob can distinguish $\mu'$ from $\mu'_X \times \mu'_Y$.
	
	Let $X = (\pi_1,\pi_3,\ldots,\pi_{2\lceil r/2 \rceil-1},A_1,\ldots,A_n)$ and $Y = (i,\pi_2,\pi_4,\ldots,\pi_{2\lfloor r/2 \rfloor},B_1,\ldots,B_n)$, where $\pi_\ell \in S_n$ and $A_k, B_k \in \{0,1\}^{Lt}$. Further, let $A_k = A_{k,1}\circ \cdots \circ A_{k,t}$ and $B_k = B_{k,1}\circ \cdots \circ B_{k,t}$ where $A_{k,\tau},B_{k,\tau} \in \{0,1\}^L$ and $\circ$ denotes concatenation.
	Alice and Bob use their common randomness to generate permutations $\sigma_{\ell,\tau}$, for $\ell \in \{0,\ldots,r\}$ and $\tau\in [t]$, uniformly and independently from $S_n$. Now let $\pi'_{\ell,\tau} = \sigma_{\ell\tau}\cdot\pi_\ell\cdot \sigma^{-1}_{\ell-1,\tau}$. Let $i'_\tau = \sigma_{0,\tau}(i)$. And let $A'_{k,\tau} = A_{\sigma_{r,\tau}(k),\tau}$ and $B'_{k,\tau} = B_{\sigma_{r,\tau}(k),\tau}$. Finally, let $X'_\tau = (\pi'_{1,\tau},\pi'_{3,\tau},\ldots,\pi'_{2\lceil r/2 \rceil-1,\tau},A'_{1,\tau},\ldots,A'_{n,\tau})$ and $Y'_\tau = (i'_\tau,\pi'_{2,\tau},\pi'_{4,\tau},\ldots,\pi'_{2\lfloor r/2 \rfloor,\tau},B'_{1,\tau},\ldots,B'_{n,\tau})$. We claim that this sequence $(X'_\tau,Y'_\tau)$ has the claimed properties.
	
	First note that the permutations $\pi'_{\ell,\tau}$ are uniform and independent from $S_n$ due to the fact that the $\sigma_{\ell,\tau}$'s are uniform and independent. Similarly $i'_\tau$'s are uniform and independent of the $\pi'_{\ell,\tau}$s. If $(X,Y)\sim \mu'_X \times \mu'_Y$ then the $A'_{k,\tau}$'s and $B'_{k,\tau}$'s are also uniform and independent of $i'$s and $\pi'$'s, estabilishing that $((X'_1,Y'_1),\ldots,(X'_t,Y'_t)) \sim \mu_X^t\times \mu_Y^t$ if $(X,Y) \sim \mu'_X \times \mu'_Y$. If $(X,Y) \sim \mu'$
	then note that $j'_\tau = \pi'_{r,\tau}(\cdots(\pi'_{1,\tau}(i'_\tau))) = \sigma_{r,\tau}(\pi_r(\cdots(\pi_1(i))))=\sigma_{r,\tau}(j)$. We thus have that $A'_{j'_\tau,\tau} = A_{j,\tau} = B_{j,\tau} = B'_{j'_\tau}$ and otherwise the $A'_{k,\tau}$'s and $B'_{k,\tau}$'s are uniform and independent. This establishes that $((X'_1,Y'_1),\ldots,(X'_t,Y'_t)) \sim \mu^t$ if $(X,Y) \sim \mu'$, and thus the proposition is proved. 

\end{proof}

\section{The Pointer Verification Problem}\label{sec:PV_prob}

When $L$ is very large compared to $n$, there are two possible natural options for trying to distinguish $\mu$ from $\mu_X \times \mu_Y$. One option is for Alice and Bob to ignore the pointers $(\pi_1,\ldots,\pi_r)$ and simply try to see if there exists $j \in [n]$ such that 
$A_j = B_j$. The second option is for Alice and Bob to ignore the $A's$ and the $B's$ while communicating and simply try to find the end of the chain of pointers $i_0=i,\ldots,i_\ell = \pi_\ell(i_{\ell-1}),\ldots,i_r$ and then check to see if $A_{i_r} = B_{i_r}$. 

The former turns out to be a problem that is at least as hard as Set Disjointness on $n$ bit inputs (and so requires $\Omega(n)$ bits of communication). The latter requires $\tilde{\Omega}(n)$ bits of communication with fewer than $r$ rounds. But combining the two lower bounds seems like a non-trivial challenge. In this section we introduce an intermediate problem, that we call the pointer verification (PV) problem, that allows us to modularly use lower bounds on the set disjointness problem and on the (small-round) communication complexity of PV, to prove that $\mu$ is indistinguishable from $\mu_X \times \mu_Y$. 

The main difference between PV and pointer chasing is that here Alice and Bob are given both a source pointer $i_0$ and a target pointer $j_0$ and simply need to decide if chasing pointers from $i_0$ leads to $j_0$. We note that the problem is definitely easier than pointer chasing in that for a sequence of $r$ pointers, Alice and Bob can decide PV in $r/2$ rounds (by ``chasing $i_0$ forward and $j_0$ backwards simultaneously''). This leads us to a bound that is weaker in the round complexity by a factor of $2$, but allows us the modularity alluded to above. Finally the bulk of the paper is devoted to proving a communication lower bound for $r/2-O(1)$ round protocols for solving PV (or rather again, an indistinguishability result for two distributions related to PV). This lower bound is similar to the lower bound of Nisan and Wigderson~\cite{NW} though the proofs are more complex due to the fact that we need to reason about settings where Alice's input and Bob's input are correlated.

We start with the definition of a distributional version of the Pointer Verification Problem and then relate it to the complexity of distinguishing $\mu$ from $\mu_X \times \mu_Y$.

\begin{definition}
	For integers $r$ and $n$ with $r$ being odd, the distributions $\dypv=\dypv(r,n)$ and $\dnpv=\dnpv(r,n)$  are supported on $((S_n^{\lceil r/2 \rceil}) \times ([n]^2 \times S_n^{\lfloor r/2 \rfloor})$.
	$\dnpv$ is just the uniform distribution over this domain. On the other hand, $(X,Y) \sim \dypv$ is sampled as follows: Sample $\pi_1,\ldots,\pi_r$ uniformly and independently from $S_n$ and further sample $i_0 \in [n]$ uniformly and independently. Finally let $j_0 = \pi_r(\cdots(\pi_1(i_0)))$, and let $X = (\pi_1,\pi_3,\ldots,\pi_r)$ and $Y = (i_0,j_0,\pi_2,\pi_4,\ldots,\pi_{r-1})$. 
\end{definition}

Our main theorem about Pointer Verification is the following:

\begin{theorem}
	\label{thm:pv-main}
	For every $\epsilon > 0$ and odd $r$ there exists $\beta, n_0$ such for every $n \geq n_0$, $\dypv(r,n)$ and $\dnpv(r,n)$ are $(\epsilon,(r-1)/2,n/\log^\beta n)$-indistinguishable.
\end{theorem}

The proof of Theorem \ref{thm:pv-main} is developed in the following sections and proved in Section~\ref{sec:proofs}. We now show that this suffices to prove our main theorem. First we prove in \Cref{lem:main-ind} below that $\mu$ is indistinguishable from $\mu_X \times \mu_Y$. This proof uses the theorem above, and the fact that set disjointness cannot be solved with $o(n)$ bits of communication, that we recall next.

\newcommand{\disjy}{\mathrm{Disj}_{\rm Y}}
\newcommand{\disjn}{\mathrm{Disj}_{\rm N}}

\begin{theorem}[\cite{Razborov}]
	\label{thm:razb}
	For every $\epsilon > 0$ there exists $\delta > 0$ such that for all $n$ the following holds:
	Let $\disjy$, respectively $\disjn$, be the uniform distribution on pairs $(U,V)$ with $U,V \subseteq [n]$ and $|U|=|V|=n/4$ such that $|U \cap V| =1$ (respectively $|U \cap V|=0$). Then $\disjy$ and $\disjn$ are $(\epsilon, \delta n, \delta n)$-indistinguishable to Alice and Bob, if Alice gets $U$ and Bob gets $V$ as inputs.
\end{theorem}

\begin{remark}
	We note that the theorem in \cite{Razborov} explicitly only rules out $(1-\eps_0,\Omega(n),\Omega(n))$-distinguishability of $\disjy$ and $\disjn$ for some $\eps_0 > 0$. But we note that the distinguishability gap of any protocol can be amplified in this case (even though we are in the setting of distributional complexity) since by applying a random permutation to $[n]$, Alice and Bob can simulate independent inputs from $\disjy$ (or $\disjn$) given any one input from its support. Thus an $(r,c)$ protocol that $\epsilon$-distinguishes $\disjy$ from $\disjn$ can be converted to an $(r,(c/\epsilon^2) \log (1/\eps_0))$-protocol that $(1-\eps_0)$-distinguishes $\disjy$ from $\disjn$, implying the version of the theorem above.
\end{remark}

\begin{lemma}
\label{lem:main-ind}
	There exists a positive integer $a$ such that for every $\epsilon > 0$ and odd $r$ there exists $\beta$ such for every $n$ and $L$, the distributions $\mu = \mu_{r,n,L}$ and $\mu_X \times \mu_Y$ are
	$(2\epsilon,r/2-a,n/\log^\beta n)$ indistinguishable.
\end{lemma}

\newcommand{\mmid}{\mathrm{mid}}

\begin{proof}
	We use a new distribution $\mu_\mmid$ which is a hybrid of $\mu$ and $\mu_X \times \mu_Y$ where $(X,Y) \sim \mu_\mmid$ is sampled as follows: Sample $\pi_1,\ldots,\pi_r \in S_n$ independently and uniformly. Further sample $i,j \in [n]$ uniformly and independently (of each other and the $\pi$'s). Finally sample $A_j = B_j \in \{0,1\}^L$ uniformly and $A_{-j}$ and $B_{-j}$ uniformly and independently from $\{0,1\}^{(n-1)L}$. Let $X = (\pi_1,\pi_3,\ldots,\pi_r,A_1,\ldots,A_n)$ and $Y = (i,\pi_2,\pi_4,\ldots,\pi_{r-1},B_1,\ldots,B_n)$. (So $\mu_\mmid$ does force a correlation between $A$ and $B$, but the permutations do not lead to this correlated point.)
	
	We show below that $\mu_\mmid$ and $\mu_X \times \mu_Y$ are indistinguishable to low-communication protocols (due to the hardness of Set Disjointness), while $\mu$ and $\mu_\mmid$ are indistinguishable to low-round low-communication protocols, due to Theorem~\ref{thm:pv-main}. The lemma follows by the triangle inequality for indistinguishability (which follows from the triangle inequality for total variation distance).

	We now use the fact (\Cref{thm:razb}) that disjointness is hard, and in particular $o(n)$-bit protocols cannot distinguish between $(U,V) \sim \disjy$ and $(U,V) \sim \disjn$. Note in particular that $\disjy$ is supported on pairs $(U, V)$ such that $U \cap V = \{j\}$ where $j \in [n]$ is distributed uniformly. Specifically, we have that for every $\epsilon > 0$ there exists $\delta > 0$ such that $\disjy$ and $\disjn$ are $(\epsilon,\delta n, \delta n)$-indistinguishable.
	
	We now show how to reduce the above to the task of distinguishing $\mu_\mmid$ and $\mu_X \times \mu_Y$ (using shared randomness and no communication). Alice and Bob share $W_1,\ldots,W_n \in \{0,1\}^L$ distributed uniformly and independently. Given $U \subseteq [n]$,  Alice picks $\pi_1,\pi_3,\ldots$ uniformly and independently, lets
	$A_\ell = W_\ell$ if $\ell \in U$ and samples $A_\ell \in \{0,1\}^L$ uniformly otherwise, and lets $X=(\pi_1,\pi_3,\ldots,\pi_r,A_1,\ldots,A_n)$. Similarly Bob samples $i \in [n]$ uniformly, and  $\pi_2,\pi_4,\ldots, \pi_{r-1} \in S_n$ uniformly and independently. Let $B_\ell = X_\ell$ if $\ell \in V$ and let $B_\ell$ be drawn uniformly from $\{0,1\}^L$ otherwise. Let $Y = (i,\pi_2,\pi_4,\ldots,\pi_{r-1},B_1,\ldots,B_n)$. It can be verified that $(X,Y) \sim \mu_\mmid$ if $(U,V) \sim \disjy$ and $(X,Y) \sim \mu_X \times \mu_Y$ if $(U,V) \sim \disjn$. Thus we conclude that
	$\mu_\mmid$ and $\mu_X \times \mu_Y$ are $(\epsilon,\delta n,\delta n)$-indistinguishable.

	Next we turn to the (in)distinguishability of $\mu$ vs. $\mu_\mmid$. We reduce the task of distinguishing $\dypv$ and $\dnpv$ to distinguishing $\mu$ and $\mu_\mmid$. Given an instance $(X,Y)$ of pointer verification with $X = (\pi_1,\pi_3,\ldots,\pi_r)$ and $Y = (i,j,\pi_2,\pi_4,\ldots,\pi_{r-1})$, we generate an instance $(X',Y')$ as follows: Let $W_1,\ldots,W_n$ be uniformly and independently chosen elements of $\{0,1\}^L$ shared by Alice and Bob. Alice lets $A_\ell = W_\ell$ for every $\ell$ and lets $X' = (\pi_1,\ldots,\pi_r,A_1,\ldots,A_n)$. Bob lets $B_j = W_j$ and samples $B_\ell$ uniformly and independently for $\ell \in [n]-\{j\}$, and lets $Y' = (i,\pi_2,\ldots,\pi_{r-1},B_1,\ldots,B_n)$. It can be verified that $(X',Y') \sim \mu$ if $(X,Y) \sim \dypv$ and $(X',Y')\sim \mu_\mmid$ if $(X,Y) \sim \dnpv$. It follows from Theorem~\ref{thm:pv-main} that $\mu$ and $\mu_\mmid$ are $(\epsilon,r/2-a,n/\log^\beta n)$-indistinguishable with $a = 1$.
	
	Combining the two we get that $\mu$ and $\mu_X \times \mu_Y$ are 
	$(2\epsilon,r/2 - a, n/\log^\beta n)$-indistinguishable (assuming $r/2-a < \delta n$ and $n/\log^\beta n < \delta n$).	
	\end{proof}

We are ready to prove Theorem~\ref{thm:main-crg}, which says that we cannot generate $\ell$ bits of common randomness from $\mu_{{r}, n, L}$ in $r/2 - 2$ rounds using only $\min(O(\ell), n/\log^\beta n)$ communication. 
\begin{proof}[Proof of \Cref{thm:main-crg}]
  We start with the case of odd $r$. We use the distribution $\mu = \mu_{r, n, L}$ in this case. Part (1) of the theorem which says that one can generate common randomness using an $(r+1,r+1\lceil \log n \rceil )$ protocol, follows from \Cref{lem:ub-skg}. Part (2) of Theorem~\ref{thm:main-crg} claims that using $r/2$ rounds and insufficient communication one cannot generate common randomness. This follows by combining \Cref{lem:main-ind} with \Cref{prop:t-removal} and \Cref{prop:need-dist}. In particular, let $\eta$ be the constant from \Cref{prop:need-dist} (and also \Cref{prop:priv-impos}), $\xi$ be the constant from \Cref{prop:need-dist}, and $\beta_0$ be the constant $\beta$ from \Cref{lem:main-ind} given the number of rounds $r$ and $(1-\ep)/40$ for the variational distance parameter. Finally let $\beta$ be a constant such that $\beta \geq \max\{ \beta_0, 3\eta + \log 1/(1-\ep)\}$ and $n/\log^{\beta} n  + \xi \log 1/(1-\ep) \leq n/\log^{\beta_0}n$, which is possible for sufficiently large $n$. Suppose for the purpose of contradiction that for some $\ell \in \mathbb{Z}^+$, there were a $((r-3)/2, \min\{ \eta \ell - \beta, n / \log^\beta n\})$-protocol for $(\ell, \ep)$-CRG from $\mu_{r,n,L}$. By \Cref{prop:need-dist}, there is some positive integer $t$ for which $\mu^t$ and $\mu_X^t \times \mu_Y^t$ are $((1-\ep)/10, (r-1)/2, \min \{ \eta \ell, n/\log^\beta n\} + \xi\log1/(1-\ep))$-distinguishable. But now let $\mu' = \mu_{r,n,Lt}$. Then by \Cref{prop:t-removal} and our assumption on $\beta$, $\mu'$ and $(\mu')_X \times (\mu')_Y$ are $((1-\ep)/10, (r-1)/2, n/\log^{\beta_0}n)$-distinguishable. But this contradicts \Cref{lem:main-ind}, which states that $\mu'$ and $(\mu')_X \times (\mu')_Y$ are $((1-\ep)/20, (r-1)/2, n/\log^{\beta_0} n)$-indistinguishable.

	
	For even $r$, we just use the distribution $\mu_{{r-1}, n, L}$. Part (1) continues to follow from \Cref{lem:ub-skg}. And for Part (2) we can reason as above, with the caveat that the bound on round complexity from \Cref{lem:main-ind} now is ``only'' $((r-1)-1)/2$. The additional loss from \Cref{prop:need-dist} is one more round, leading to a final lower bound of $r/2-2$.

\end{proof}

\begin{proof}[Proof of \Cref{thm:main-skg}]
	Part (1) of the theorem follows from \Cref{lem:ub-skg}. Part (2) follows from Part (2) of \Cref{thm:main-crg} since SKG is a strictly harder task.
\end{proof}

\section{Proof of {\protect Theorem~\ref{thm:pv-main}}}
\label{sec:proofs}

In this section we prove our main technical theorem \Cref{thm:pv-main} showing that the distributions $\dypv(r,n)$ and $\dnpv(r,n)$ are indistinguishable to $(r/2- O(1), n/\poly\log n)$-protocols (i.e., $r/2-O(1)$ round protocols communicating $n/\poly\log n$ bits). We start with some information-theoretic preliminaries.

\subsection{Preliminaries: Information-Theoretic Inequalities}
We introduce here some simple information theoretic inequalities that we use in our proofs. Pinsker's inequality gives an upper bound on the total variation distance between two distributions in terms of their KL-divergence. Recall that the KL-divergence between two discrete distributions $P$ and $Q$ is defined as $D_{KL}(P || Q) = \sum_{x \in \Omega} P(x) \log(P(x)/Q(x))$ where $\Omega$ is the support of $P$. 
\begin{theorem}[Pinsker's Inequality]
	Let $P$ and $Q$ be two distributions defined on the universe $U$. Then,
	$$\Delta(P,Q) \leq \sqrt{\frac{D_{KL}(P || Q)}{2}},$$
	where $\Delta(P,Q) \in [0,1]$ is the total variation distance. 
\end{theorem}

In the case that $Q$ is uniform, Theorems \ref{thm:ho} and \ref{thm:cover} below give a sort of reverse inequality to Pinsker's inequality. In particular, when $Q = U_M$, the uniform distribution on $[M]$, then $D_{KL}(P || Q) = \log(M) - H(P) = H(Q) - H(P)$, so an upper bound on $H(Q) - H(P)$ corresponds to an upper bound on $D_{KL}(P || Q)$. A similar line of reasoning applies to the case that $Q$ is approximately uniform. 
\begin{theorem}[\cite{ho_interplay_2010}, Theorem 6]
	\label{thm:ho}
	Suppose that $P,Q$ are distributions on $[M]$, for some $M \in \N$. If moreover $\Delta(P,Q) \leq \ep$, then
	$$
	|H(P) - H(Q)| \leq \begin{cases}
	h\left( \ep \right) + \ep \log(M-1), \quad 0 < \ep \leq \frac{M-1}{M}\\
	\log(M), \quad \ep \geq \frac{M-1}{M},
	\end{cases}
	$$
	where $h(\cdot)$ denotes the binary entropy.
\end{theorem}
We remark that \cite{ho_interplay_2010} showed that the above inequality is tight, i.e.,~that there are distributions $P,Q$ supported on $[M]$ such that $\Delta(P,Q) \leq \ep$ and $P,Q$ attain the above upper bound for all values of $\ep$. 

The following slightly weaker theorem is also well-known:
\begin{theorem}[\cite{cover_elements_2006}, Theorem 17.3.3]
	\label{thm:cover}
	Suppose that $P,Q$ are distributions on $[M]$ and $\Delta(P,Q) \leq \ep \leq 1/2$. Then
	$$
	|H(P) - H(Q)| \leq \ep \cdot \log\left( \frac{M}{\ep} \right).
	$$
\end{theorem}

\subsection{A Reformulation of \Cref{thm:pv-main}}

In this section we state \Cref{lem:pv-hybrid} which is a slight reformulation of \Cref{thm:pv-main}
and then show how \Cref{thm:pv-main} follows from \Cref{lem:pv-hybrid}. The remaining subsections will then be devoted to the proof of \Cref{lem:pv-hybrid}.

We first introduce some additional notation for the pointer verification problem. For $s < t$, let $\pi_s^t = \pi_t \circ \pi_{t-1} \circ \cdots \pi_s$ and $(\pi^{-1})_t^s = \pi_s^{-1} \circ \cdots \circ \pi_t^{-1}$. Also let $i_s = \pi_1^s(i_0), j_s = (\pi^{-1})_r^{r-s+1}(j_0)$. Then over the distribution $\dypv$, $j_r = i_0$ and $i_r = j_0$ with probability 1. We also write $\pi_A = (\pi_1, \pi_3, \ldots, \pi_r)$ and $\pi_B = (\pi_2, \pi_4, \ldots, \pi_{r-1})$. Recall that Alice holds the permutations $\pi_A$ while Bob holds the permutations $\pi_B$.
For technical reasons, in this section, we consider protocols that get inputs sampled from a single ``mixed'' distribution, $\dmpv = \frac12(\dypv + \dnpv)$ and outputs a bit (last bit of the transcript) that aims to guess whether the input is a YES input to Pointer Verification ($\pi_1^r(i_0) = j_0$) or a NO input ($\pi_1^r(i_0) \ne j_0$). The success of a protocol is the probability with which this bit is guessed correctly. These terms are formally defined below.

\begin{definition}
	For any odd integer $r$ and any integer $n$, the distribution $\dmpv = \dmpv(r,n)$ is supported on $(S_n^{\lceil r/2 \rceil}) \times ([n]^2 \times S_n^{\lfloor r/2 \rfloor})$, and is defined by drawing $\dnpv(r,n)$ with probability 1/2 and drawing $\dypv(r,n)$ with probability 1/2.
	
	A protocol $\Pi$ is said to achieve {\em success} on a pair of inputs drawn from $\dmpv$ if the last bit of the transcript of $\Pi$, which we take as the output bit, is 1 if and only if $\pi_1^r(i_0) = j_0$.
\end{definition}

In \Cref{lem:pv-hybrid} we show that Alice and Bob cannot achieve success with probability significantly greater than 1/2 when their inputs are drawn from $\dmpv$. \Cref{thm:pv-main} follows fairly easily from \Cref{lem:pv-hybrid}.

\begin{lemma}
	\label{lem:pv-hybrid}
	For every $\ep > 0$ and every $r$, there exists $\beta, n_0$ such that for every $n \geq n_0$ the following holds: Every  $((r+1)/2, n/\log^\beta(n))$ protocol on $\dmpv$ achieves success with probability at most $1/2 + \ep$.
\end{lemma}

We defer the proof of \Cref{lem:pv-hybrid} but first show how \Cref{thm:pv-main} follows from it. 


\begin{proof}[Proof of Theorem \ref{thm:pv-main}]
		
\Cref{lem:pv-hybrid} gives that there exists $\beta, n_0$ such that for every $n \geq n_0$, no $((r+1)/2, n/\log^\beta(n))$ protocol $\Pi$ on $\dmpv(r,n)$ achieves success with probability greater than $1/2 + \ep/4$. Suppose for the purpose of contradiction that there were an $((r-1)/2, n/\log^\beta(n)-1)$ protocol that $\ep$-distinguishes $\dypv(r,n)$ and $\dnpv(r,n)$. Then by the definition of $\ep$-distinguishability, by modifying this protocol to output an extra bit (which we interpret as the output bit), we get an $((r+1)/2, n/\log^\beta(n))$ protocol $\Pi'$ which outputs 1 with probability $p_Y$ when the inputs are drawn from $\dypv(r,n)$ and which outputs 1 with probability $p_N$ when the inputs are drawn from $\dnpv(r,n)$, where $p_Y \geq p_N + \ep$. Therefore, $\Pi'$ has probability of success of at least $1/2 + \ep/2$ when the inputs are drawn from $\dmpv(r,n)$, which contradicts \Cref{lem:pv-hybrid}.

\end{proof}

\subsection{Proof of the Main Lemma (\Cref{lem:pv-hybrid}): Setting up the Induction}

Our approach to the proof of \Cref{lem:pv-hybrid} is based on the ``round-elimination'' approach of \cite{NW}. Roughly, given inputs drawn from $\dmpv(n,r)$, the approach here is to show that after a single message $m = m(\pi_A)$ from Alice to Bob, Alice and Bob are still left with essentially a problem from $\dmpv(n,r-2)$ (with their roles reversed). Note that the distribution of $(\pi_2,\ldots,\pi_{r-1}; i_1,j_1)$, where $i_1 = \pi_1(i_0)$ and $j_1 = \pi_r^{-1}(j_0)$, is exactly $\dmpv(n,r-2)$ (with the roles of Alice and Bob switched). The crux of the \cite{NW} approach is to show that this roughly remains the case even when conditioned on the message $m = m(\pi_A)$ sent in the first round. If implemented correctly, this would lead to an inductive strategy for proving the lower bound, with the induction asserting that an additional $(r-2)/2$ rounds of communication do not lead to non-trivially high success probability. 
Of course the distributions of the inputs after conditioning on $m$ are not exactly the same as $\dmpv(n,r-2)$. Bob can definitely learns a lot of information about Alice's input $\pi_A$ 	from $m$.
So the inductive hypothesis needs to deal with distributions that retain some of the features of  $\dmpv(n,r)$ while allowing Alice and Bob to have a fair amount of information about each others inputs. In \Cref{def:dmixdef} we present the exact class of distributions with which we work. While most of the properties are similar to those used in \cite{NW} the exact definition is not immediate since we need to ensure that the bit ``Is $\pi_1^r(i_0) = j_0$" is not determinable even after a few rounds of communication. (In our definition, Item~\ref{item:indkey} in particular is the non-trivial ingredient.) In \Cref{lem:simulation} we then show that this definition supports induction on the number of rounds of communication. Finally in \Cref{lem:base-case} we show that the base-case of the induction with $r=1$ does not achieve non-trivial success probability. The proofs of \Cref{lem:base-case} and \Cref{lem:simulation} are deferred to \Cref{ssec:base-case} and \Cref{ssec:simulation} respectively. We conclude the current section with a proof of \Cref{lem:pv-hybrid} assuming these two lemmas.
	
We start with our definition of the class of ``noisy'' distributions, containing $\dmpv$. In particular, for $n,r, \delta, C$ satisfying $0 \leq \delta < 1$ and $0 \leq C < n$, we define the class of distributions $\Dmpv(n,r,\delta,C)$ in Definition \ref{def:dmixdef} below.


\begin{definition}
	\label{def:dmixdef}
	The set of {\em noisy} distributions, denoted $\Dmpv(n,r,\delta,C)$, consists of those distributions $D$ supported on $((S_n^{\lceil r/2 \rceil}) \times ([n]^2 \times S_n^{\lfloor r/2 \rfloor})$, satisfying the following properties. If we denote a sample from $D$ as $(i_0, j_0, \pi_1, \ldots, \pi_r)$, then
	\begin{enumerate}
		\item \label{item:dmpv1} \begin{enumerate}\item  $H(i_0 | \pi_1, \ldots, \pi_r) \geq \log(n) - \delta$ \item $H(j_0 | \pi_1, \ldots, \pi_r) \geq \log(n) - \delta$.\end{enumerate}
		\item $H(\pi_1, \ldots, \pi_r) \geq r \log(n!) - C$.
		\item \label{item:indkey} \begin{enumerate} \item $H(\One[\pi_1^r(i_0) = j_0] | i_0, \pi_1, \ldots, \pi_r) \geq 1 - \delta$. \item $H(\One[\pi_1^r(i_0) = j_0] | j_0, \pi_1, \ldots, \pi_r) \geq 1 - \delta$.\end{enumerate}
		\item \begin{enumerate} \item $H(j_0 | i_0, \pi_1, \ldots, \pi_r, \pi_1^r(i_0) \neq j_0) \geq \log(n) - \delta$. \item $H(i_0 | j_0, \pi_1, \ldots, \pi_r, \pi_1^r(i_0) \neq j_0) \geq \log(n) - \delta$.\end{enumerate}
		\item \label{item:dmpv5}For all odd $1 \leq t \leq r$, the following conditional independence properties hold. For all $i_0', \ldots, i_t', j_0', \ldots, j_t' \in [n]$, $\pi_{t + 2}', \pi_{t+4}', \ldots, \pi_{r-t-1}' \in S_n$,
		\begin{eqnarray}
		\pi_A \cap (\pi_1, \ldots, \pi_t, \pi_{r-t+1}, \ldots, \pi_r) \indep \pi_B &|& (i_0, \ldots, i_t) = (i_0', \ldots, i_t'), (j_0, \ldots, j_t) = (j_0', \ldots, j_t'), \nonumber\\
		&& (\pi_{t+2}, \pi_{t+4}, \ldots, \pi_{r-t-1}) = (\pi_{t+2}', \pi_{t+4}', \ldots, \pi_{r-t-1}')\nonumber.
		\end{eqnarray}
		and for all even $t$, $0 \leq t \leq r$, $i_0', i_1', \ldots, i_t', j_0', j_1', \ldots, j_t' \in [n]$, $\pi_{t+2}', \pi_{t+4}', \ldots, \pi_{r-t-1}' \in S_n$,
		\begin{eqnarray}
		\pi_B \cap (\pi_2, \ldots, \pi_t, \pi_{r-t+1}, \ldots, \pi_{r-1}) \indep \pi_A &|& (i_0, \ldots, i_t) = (i_0', \ldots, i_t'), (j_0, \ldots, j_t) = (j_0', \ldots, j_t'),\nonumber\\
		&& (\pi_{t+2}, \pi_{t+4}, \ldots, \pi_{r-t-1}) = (\pi_{t+2}', \pi_{t+4}', \ldots, \pi_{r-t-1}')\nonumber.
		\end{eqnarray}
	\end{enumerate}
	The set of {\em noisy-on-average} distributions, $\Dmpvp(n,r,\delta,C)$, consists of those distributions $D^+$ supported on $((S_n^{\lceil r/2 \rceil}) \times ([n]^2 \times S_n^{\lfloor r/2 \rfloor}) \times \mathcal{Z}$ where $\mathcal{Z}$ is some finite set and a sample $(i_0, j_0, \pi_1, \ldots, \pi_r,Z) \sim D^+$ satisfies Properties~(\ref{item:dmpv1})-(\ref{item:dmpv5}) when all quantities above are additionally conditioned on $Z$. (In particular the conditional entropies are additionally conditioned on $Z$ and the independences hold when conditioned on $Z$.)
\end{definition}

We first state a version of \Cref{lem:pv-hybrid} for every distribution $D \in \Dmpv(n,r,\delta, C)$, for sufficiently small $\delta, C$. We also show that $\dmpv$ belongs to this set for the permissible $\delta, C$, and thus \Cref{lem:pv-hybrid-strong} implies \Cref{lem:pv-hybrid}. 

\begin{lemma}
	\label{lem:pv-hybrid-strong}
	For every $\ep>0$ and odd $r$, there exists $\beta$ and $n_0$ such that for every $n \geq n_0$, and every $D \in \Dmpv(n,r,1/\log^{\beta} n, n/ \log^{\beta} n)$ it is the case that every $((r+1)/2, n/\log^\beta(n))$-protocol achieves success with probability at most $1/2 + \ep$ on $D$.
\end{lemma}

\begin{remark}
	In the lemma statement we have suppressed the dependence of $\beta$ on $r$. (The dependence of $\beta$ on $\epsilon$ is minimal. Essentially only $n_0$ is affected by $\epsilon$.) A careful analysis (based on the remarks after \Cref{lem:base-case} and \Cref{lem:simulation}) yields that $\beta$ grows exponentially in $r$, though we omit the simple but tedious bookkeeping.
\end{remark}

The proof of \Cref{lem:pv-hybrid-strong} is via induction on $r$; the below lemma gives the main inductive step, which says that if one cannot solve the pointer verification problem with $r-2$ permutations then one cannot hope to solve the problem on $r$ permutations even with an additional round of (not too long) communication.

\begin{lemma}[Inductive step]\label{lem:simulation}
	For every $\ep_1 > \ep_2 > 0$, odd $r$ and $\beta_2$ there exists $\beta_1$ and $n_0$ such that for every $n \geq n_0$ the following holds:
	Suppose there exists $D \in \Dmpv(n,r,1/\log^{\beta_1} n, n/\log^{\beta_1}n)$ and an $((r+1)/2,n/\log^{\beta_1} n)$-protocol $\Pi$ that achieves success $1/2 + \ep_1$ on $D$. 
	Then there exists $\tilde{D} \in \Dmpv(n,r-2,1/\log^{\beta_2} n, n/\log^{\beta_2}n)$ and an $((r-1)/2,n/\log^{\beta_2} n)$-protocol $\tilde{\Pi}$ that achieves success $1/2 + \ep_2$ on $\tilde{D}$.

\end{lemma}

\begin{remark} A careful analysis of the proof yields that $\beta_2$ grows linearly with $\beta_1$ with some mild conditions on $n_0$ and $\ep_1-\ep_2$.
\end{remark}

The proof of \Cref{lem:pv-hybrid-strong} proceeds by using Lemma \ref{lem:simulation} repeatedly, to reduce the case with general $r$ to the case with $r=1$. In the case $r=1$, Alice is given one permutation $\pi_1$, Bob is given indices $i_0, j_0$, and Alice can communicate one message to Bob, who has to then decide whether $\pi_1(i_0) = j_0$ or not. 
The next lemma, \Cref{lem:base-case}, asserts that the pointer verification problem with $r=1$ cannot be solved in one round with less than $n/\log^{O(1)}(n)$ communication. In fact the lemma is a stronger one, where we show that if all the statements hold conditioned on a random variable $Z$, then the entropy of the indicator of the outcome is large even when conditioned on $Z$. Setting $Z$ to be a constant immediately yields the base case of the induction with $r=1$, as noted in \Cref{cor:base-case}. (We note that we need the stronger version stated in the lemma, i.e., with a general random variable $Z$, in the proof of \Cref{lem:simulation}.)

\begin{lemma}[Base case]\label{lem:base-case}
	There exists $0 < \ep_1^* < 1$ and $\ep_2^*$ such that  for every $\tilde{\beta}$ there is $n_0$ such that the following holds for every $n \geq n_0$. Let $\beta = (\tilde{\beta}+\ep_2^*)/\ep_1^*$, $\delta = 1/\log^{\beta}n$ and $C,C' = n/\log^{\beta} n$. Suppose $(i,j,\pi,Z)$ are drawn from a distribution $D$, where $Z$ is a random variable that takes on finitely many values, such that the following properties hold:
	\begin{enumerate}
		\item \label{item:base-case-1} $H(i | \pi, Z) \geq \log(n) - \delta$.
		\item $H(\pi | Z) \geq \log(n!)- C$.
		\item $H(\One[\pi(i) = j] | \pi, i, Z) \geq 1-\delta$.
		\item \label{item:base-case-4} $H(j | \pi, i, \One[\pi(i) \neq j], Z) \geq \log(n) - \delta$.
	\end{enumerate}
    Then for every deterministic function $m = m(\pi, Z)$ with $m \in \{0,1\}^{C'}$ we have the following:
	\begin{align}
	H(\pi(i) | i,m, Z)& \geq \log n - 1/\log^{\tilde{\beta}}n\label{eq:lem-part1} \\
	\mbox{~and~~} H(\One[\pi(i) = j] | m,i,j, Z)& \geq 1- 1/\log^{\tilde{\beta}}n. \label{eq:lem-part2}
	\end{align}
	
\end{lemma}

\begin{remark} The proof shows that $\beta$ grows linearly with $\tilde \beta$ provided that $n_0$ is sufficiently large (as a function of $\tilde \beta$).
\end{remark}

\begin{corollary}
	\label{cor:base-case}
	For every $\ep>0$, there exists $\beta_0$ and $n_0$ such that for every $n \geq n_0$, and every $D \in \Dmpv(n,1,1/\log^{\beta_0} n, n/ \log^{\beta_0} n)$ it is the case that every $(1, n/\log^{\beta_0}(n))$-protocol achieves success with probability at most $1/2 + \ep$ on $D$.
\end{corollary}

\begin{proof}
	Recall that a 1-round distribution $D \in \Dmpv(n,1,\delta,C)$ is supported on triples $(\pi,i,j)$ and the goal is to determine if $\pi(i) = j$. We apply \Cref{lem:base-case} with $Z = 0$ (i.e., a constant). Given $\ep>0$ we let $\tilde{\beta} = 1$ and let $\beta$ be as given by \Cref{lem:base-case}. Further let $n'_0$ denote the lower bound on $n$ returned by \Cref{lem:base-case}. Let $\epsilon'$ be such that a binary variable of entropy at least $1 - \ep'$ is Bernoulli with bias in the range $[1/2-\ep,1/2+\ep]$ ($\ep'= O(\ep^2)$ works). We prove the claim for $\beta_0 = \beta$ and $n_0 = \max\{n'_0,2^{1/(\epsilon')}\}$ (so that $\log^{\tilde{\beta}} n \leq \ep'$ for all $n \geq n_0$). 
	
	By definition of $\Dmpv(n,1,1/\log^{\beta_0} n, n/\log^{\beta_0}n)$, we have that for $(\pi,i,j) \sim D$, the conditions (\ref{item:base-case-1})-(\ref{item:base-case-4}) of \Cref{lem:base-case} hold for $(\pi,i,j,Z)$ (where $Z$ is simply the constant $0$). Thus \Cref{lem:base-case} asserts that $H(\One[\pi(i) = j] | m,i,j, Z) \geq 1- 1/\log^{\tilde{\beta}}n \geq 1- \ep'$ for any message $m = m(\pi) \in \{0,1\}^{C'}$ sent by Alice. Let $\Pi(m,i,j)$ denote the output bit of the protocol output by Bob. Since this is a deterministic function of $m,i,j$ we have, by the data processing inequality,  that $H(\One[\pi_1(i_0) = j_0] | \Pi(m_1, i_0, j_0)) \geq 1 - \ep'$. By the choice of $\ep'$ and Jensen's inequality (to average over the conditioning on $\Pi(m,i,j)$) we have that 
$$
\Pr \left[ \One[\pi(i) = j] = \Pi(m,i,j) \right] \leq 1/2 + \ep,
$$
which verifies that the success probability of the protocol $\Pi$ is at most $1/2 + \ep$ as asserted.	
\end{proof}

Armed with \Cref{lem:simulation} and \Cref{cor:base-case} we are now ready to prove \Cref{lem:pv-hybrid-strong}.

\begin{proof}[Proof of \Cref{lem:pv-hybrid-strong}]
	We prove the lemma by induction on $r$. If $r=1$, then \Cref{cor:base-case} gives us the lemma. Assume now that the lemma holds for all odd $r' < r$. In particular, let $\beta_{r-2}$ and $n_{0,r-2}$ be the parameters given by the lemma for $r-2$ rounds and parameter $\ep/2$. We now apply \Cref{lem:simulation} with parameters $\ep_1 = \ep$, $\ep_2 = \ep/2$, $r$ rounds and $\beta_2 = \beta_{r-2}$. Let $n'_0$ and $\beta_1$ be the parameters given to exist by \Cref{lem:simulation}. We verify the inductive step with $n_{0,r} = \max\{n_{0,r-2},n'_0\}$ and $\beta_r = \beta_1$. Fix $D \in \Dmpv(n,r,1/\log^{\beta_r}n, n/\log^{\beta_r}n)$ and assume for contradiction that an $((r+1)/2,n/\log^{\beta_r}n)$-protocol achieves success $1/2 + \ep$ on $D$. Then by \Cref{lem:simulation} we have that there exists $\tilde{D} \in \Dmpv(n,r-2,1/\log^{\beta_{r-2}}n, n/\log^{\beta_{r-2}}n)$ and an $((r-1)/2,n/\log^{\beta_{r-2}}n)$-protocol $\tilde{\Pi}$ that achieves success $1/2 +\ep/2$ on $\tilde{D}$, which contradicts the inductive hypothesis.
	\end{proof}

We finally show how \Cref{lem:pv-hybrid} follows from \Cref{lem:pv-hybrid-strong} (which  amounts to verifying the $\dmpv$ satisfies the requirements of membership in $\Dmpv$ for appropriate choice of parameters). 

\begin{proof}[Proof of \Cref{lem:pv-hybrid}]
We claim that for each odd integer $r$, $\dmpv(r,n) \in \Dmpv(n,r, 2/n, 0)$ for sufficiently large $n$. To verify this, note that if $(\pi_1, \ldots, \pi_r, i_0, j_0)$ are drawn from $\dmpv(r,n)$, then
\begin{enumerate}
	\item $H(i_0 | \pi_1, \ldots, \pi_r) = H(j_0 | \pi_1, \ldots, \pi_r) = \log(n)$.
	\item $H(\pi_1, \ldots, \pi_r) = r \cdot \log(n!)$.
	\item $H(\One[\pi_1^r(i_0) = j_0] | i_0, \pi_1, \ldots, \pi_r) = H(\One[\pi_1^r(i_0) = j_0] | j_0, \pi_1, \ldots, \pi_r) = h(1/2 + 1/(2n)) \geq 1-1/n^2$.
	\item $H(j_0 | i_0, \pi_1, \ldots, \pi_r, \pi_1^r(i_0) \neq j_0) = H(i_0 | j_0, \pi_1, \ldots, \pi_r, \pi_1^r(i_0) \neq j_0) = \log(n-1) \geq \log(n) - 2/n$, for sufficiently large values of $n$.
	\item To verify the conditional independence properties (5) from Definition \ref{def:dmixdef}, first fix any odd $t$ such that $1 \leq t \leq r$, and pick any $i_0', \ldots, i_t', j_0', \ldots, j_t' \in [n]$ and $\pi_{t+2}',\pi_{t+4}',\ldots, \pi_{r-t-2}' \in S_n$. Given that
	$$
	\{ (i_0, \ldots, i_t) = (i_0', \ldots, i_t'), (j_0, \ldots, j_t) = (j_0', \ldots, j_t'), (\pi_{t+2}, \pi_{t+4}, \ldots, \pi_{r-t-1}) = (\pi_{t+2}', \pi_{t+4}', \ldots, \pi_{r-t-1}') \},
	$$
	and regardless of the choice of $\pi_B$, note that the permutations in $\pi_A \cap (\pi_1, \ldots, \pi_t, \pi_{r-t-1}, \ldots, \pi_r)$ are uniformly random subject to $\pi_{s}(i_{s-1}') = i_s'$ for $s \in \{1, 3, \ldots, t\}$ and $\pi_{r-s+1}^{-1}(j_s') = j_{s-1}'$ for $s \in \{1, 3, \ldots, t\}$. A similar argument verifies the analogous statement for even $t$.
\end{enumerate}
In particular, it follows that for every  $\beta > 0$ and every odd $r$, for sufficiently large $n$, we have that $\dmpv(r,n) \in \Dmpv(n, r, 1/\log^{\beta}(n), n/\log^{\beta}(n))$, and in particular this holds for the parameter $\beta$ guaranteed to exist by \Cref{lem:pv-hybrid-strong}.
The lemma now follows immediately from the conclusion of \Cref{lem:pv-hybrid-strong}, which asserts that every $((r+1)/2, n/\log^\beta(n))$-protocol achieves success with probability at most $1/2 + \ep$ on $D$.
\end{proof}

Thus the main lemma is proved assuming \Cref{lem:base-case} and \Cref{lem:simulation}. In the rest of this section we prove these two lemmas.

\subsection{The Base Case: Proof of {\protect Lemma~\ref{lem:base-case}}}
\label{ssec:base-case}

In the following we will fix
$\beta$ and argue that if $\tilde{\beta} \leq \ep_1^* \cdot \beta - \ep_2^*$ then the conditions (\ref{eq:lem-part1}) and (\ref{eq:lem-part2}) of \Cref{lem:base-case} hold.
Specifically we will prove (\ref{eq:lem-part1}) first and then derive  (\ref{eq:lem-part2}) as a consequence. For (\ref{eq:lem-part1}), we will first bound $H(\pi(i) | i)$ when  $\pi$ is a nearly uniform {\it function} instead of a nearly random {\it permutation}, and then extend it to case that $\pi$ is a nearly uniform {permutation}. Then using this result, we will bound $H(\pi(i) | i,m)$, where $m$ is a short message that depends on $\pi$.

In the below Lemma \ref{lem:noisyimfn}, we will take $i \in [k]$ and $\pi : [k] \rightarrow [n]$ to be a nearly uniformly random function. We allow that $k \neq n$ in order to deal with the case that $\pi$ is a nearly uniformly random permutation later on (in our application we will always have $k \leq n$).

\begin{lemma}
	\label{lem:noisyimfn}
	For every $k,n \in \Z_+$ and every $\delta,C \in \R_+$ the following holds: Suppose $(i,\pi)$ are drawn from a distribution $D$ such that the resulting random variables, $i \in [k],\pi : [k]\to [n]$ have the following properties:
	\begin{enumerate}
		\item $H(i | \pi) \geq \log(k) - \delta$, with $\delta \in [1/n,1/8)$.
		\item $H(\pi) \geq k\log n- C$, with $C \leq k$.
        \end{enumerate}
        Then
	$$
	H(\pi(i) | i) \geq \log(n) - \frac Ck - 2 \sqrt{2\delta} \log(n).
	$$
\end{lemma}

\begin{proof}
Let $D$ be the joint distribution on $(\pi,i)$ that satisfies (1),(2) and let $D_i,D_{\pi}$ be its marginals on $i$ and $\pi$ respectively. Unless specified, all the following probability statements are with respect to $D$. Let $U_k$ denote the random variable that is uniform on $[k]$. 

We will first make a few observations and then bound $H(\pi(i) | i)$. Firstly, since $H(i) \geq \log k - \delta$, by Pinsker's inequality, we have that,
\begin{equation}\label{pins-i}
\Delta(D_i, U_k) = \frac 12 \sum_{i'=1}^k|\Pr[i = i'] - 1/k|  \leq \sqrt{\delta/2}.	
\end{equation}
	
Let $D_\pi \otimes D_i$ denote the joint distribution over $(\pi,i)$, where $\pi$ and $i$ are independently drawn from their marginals $D_\pi$ and $D_i$ respectively. By Pinsker's inequality, we have that, $$\Delta(D, D_\pi \otimes D_i) \leq \sqrt{I(\pi; i)/2} \leq \sqrt{\delta/2}.$$ It then follows that,
\begin{equation}
\label{eq:abovedelta}
\sum_{i' \in [k], j' \in [n]} \left|\Pr[\pi(i') = j', i=i'] - \Pr[\pi(i') = j'] \cdot \Pr[i=i'] \right| \leq \sqrt{2\delta}.
\end{equation}
Now, for each $i' \in [k]$, define,
\begin{eqnarray}
\ep_{i'} &=& \sum_{j' \in [n]} \left| \Pr[\pi(i') = j', i=i'] - \Pr[\pi(i') = j'] \cdot \Pr[i=i'] \right|,\nonumber
\end{eqnarray}
so that $\sum_{i' \in [k]} \ep_{i'} \leq \sqrt{2\delta}$. We get that
\begin{align}
&\Delta((\pi(i') | i= i'), \pi(i')) = \frac{1}{2}\sum_{j' \in [n]} \left| \Pr[\pi(i') = j' | i=i'] - \Pr[\pi(i') = j'] \right| = \frac{\ep_{i'}}{2\Pr[i = i']} \nonumber,
\end{align}
which by Theorem~\ref{thm:ho} then gives,
\begin{align}
\left| H(\pi(i') | i=i') - H(\pi(i')) \right| \leq h \left(\frac{\ep_{i'}}{2\Pr[i=i']} \right) + \left(\frac{\ep_{i'}}{2\Pr[i=i']}\right) \log (n-1) := \beta_{i'}  \label{pins-pi}.		
\end{align}

We have that
\begin{eqnarray}
	H(\pi(i) | i) &=& \sum_{i' \in [k]} \Pr[i = i'] \cdot H(\pi(i) | i = i') \nonumber \\
	&\geq & \sum_{i'} \Pr[i = i'] (H(\pi(i')) - \beta_{i'}) \label{eq1} \\
	& \geq & \sum_{i'} \frac{1}{k} H(\pi(i')) -\sqrt{\delta/2}\log n - \sum_{i'} \Pr[i = i'] \beta_{i'},	\label{eq2}
\end{eqnarray}
where (\ref{eq1}) follows from (\ref{pins-pi}), and (\ref{eq2}) follows from (\ref{pins-i}) and the fact that $H(\pi(i')) \leq \log n$.

Using the chain rule for entropy we get that
\begin{equation}
\label{eq:hp}
\log n - C/k   \leq \frac{1}{k}H(\pi) = \frac{1}{k}\sum_{i'=1}^k H(\pi(i') | \pi(\{1, \ldots, i'-1\})) \leq \frac{1}{k}\sum_{i'=1}^k H(\pi(i')).
\end{equation}
Recall that $\sum_{i'} \ep_{i'} \leq \sqrt{2\delta}$ and we have that $h(\sum_{i'} \ep_{i'}) \leq h(\sqrt{2\delta})$, since $\delta < 1/8$. Since the binary entropy function $h(\cdot)$ is concave, by Jensen's inequality, we have that, 
\begin{align}\label{beta-ineq}
	\sum_{i'=1}^k \Pr[i = i']\beta_{i'} &= \sum_{i'} \Pr[i = i'] h \left(\frac{\ep_{i'}}{2\Pr[i=i']} \right) + \sum_{i'} \Pr[i = i']\left(\frac{\ep_{i'}}{2\Pr[i=i']}\right) \log (n-1) \nonumber\\
	& \leq   h\left(\sum_{i'} \Pr[i=i'] \cdot \frac{\ep_{i'}}{2\Pr[i=i']} \right) + \sqrt{\delta/2}\log n \nonumber \\
	&\leq h\left(\sqrt{\delta/2} \right) + \sqrt{\delta/2}\log n.
\end{align}

Note that $h(x) \leq 2x\log(1/x)$ for $x \rightarrow 0$, so $h(\sqrt{\delta/2}) \leq \sqrt{2\delta}\log n$. Using this, and plugging (\ref{eq:hp}) and (\ref{beta-ineq}) into (\ref{eq2}), we get that
\begin{equation*}
	H(\pi(i) | i) \geq \log n -\frac{C}{k} - 2\sqrt{\delta/2}\log n \geq \log(n) - \frac Ck - 2\sqrt{2\delta} \log(n).\qedhere
\end{equation*}	
\end{proof}

Now we are ready to prove an analogous lemma for random permutations instead of random functions. We note that we cannot replicate the proof above since for a typical $i'$ the conditional entropy $H(\pi(i') | \pi(\{1, \ldots, i'-1\}))$ is actually $\log n - \Theta(1)$ and this $\Theta(1)$ loss is too much for us. In the proof below we condition instead on $i$ being contained in some smaller set $S \subseteq [n]$, with $|S| = k = o(n)$, where $S$ itself is randomly chosen. This ``conditioning'' turns out to help with the application of the chain rule and this allows us to reproduce a bound that is roughly as strong as the bound above.

\begin{lemma}
	\label{lem:noisyimperm}
	There exists constants $\ep_1^*>0, \ep_2^*$ such that for every $\beta$ there exists $n_0$ such that for all $n \geq n_0$ the following holds:
	Suppose $i\in [n]$, $\pi \in S_n$ are random variables such that:
	\begin{enumerate}
		\item $H(i | \pi) \geq \log(n) - \delta$, with $\delta \in [1/n,1/\log^{\beta}n]$.
		\item $H(\pi ) \geq \log(n!) - C$, with $C \leq n/\log^\beta(n)$.
	\end{enumerate}
	Then 
	$$H(\pi(i) | i) \geq \log n - 1/\log^{\tilde{\beta}}n,$$ where $\tilde{\beta} = \ep_1^*\cdot \beta - \ep_2^*$.
\end{lemma}

\newcommand{\cE}{\mathcal{E}}
\begin{proof}
  We will prove the lemma with $\ep_1^* = 1/16, \ep_2^* = 4$. Note that for $\beta \leq 8$, $\tilde \beta = \ep_1^* \beta - \ep_2^* \leq -3$, so by non-negativity of entropy, the lemma statement follows immediately. We therefore assume $\beta > 8$ for the remainder of the proof.
	
	Let $D$ be the distribution of $(\pi, i)$ given in the lemma statement, where $D_{\pi}, D_{i}$ are its marginals on $i,\pi$ respectively. Let $k$ be a parameter to be fixed later. 
	We start by defining a joint distribution $D'$ on triples $(\pi,i,S)$ with $\pi \in S_n$ and $i \in S \subset [n]$, $|S| = k$ that satisfies the condition that its marginal on $(\pi,i)$ equals $D$ while at the same time the distribution of $(\pi,i)$ conditioned on $S = S'$ when $(\pi,i,S) \sim D'$ is the same as the distribution of $(\pi,i) \sim D$ conditioned on $i \in S'$. $D'$ is defined as follows:
	
	  Let $D_S$ be the distribution of $(\pi, i)$, conditioned on $i \in S$. Now let  $\mathcal{E}$ be the distribution over subsets $S \subset [n]$ of size $k$ where the probability of $\Pr_{S\sim \cE}[S = S'] = \frac{\sum_{i' \in S'}\Pr_D[i = i']}{{n-1 \choose k-1}}$. Now define the joint distribution $D'$ of $(\pi, i, S)$ of $\pi' \in S_n, i'\in S' \subset [n], |S'| = k$ so that
\begin{eqnarray*}
\Pr_{D'}[\pi = \pi', i=i', S=S'] &=& \Pr_\mathcal{E}[S=S'] \cdot \Pr_{D}[\pi = \pi', i=i' | i \in S']\\
&=& \Pr_\mathcal{E}[S=S'] \cdot \Pr_{D_{S'}} [\pi = \pi', i=i'].
\end{eqnarray*}
We claim that the marginal distribution of $(\pi, i)$, where $(\pi, i, S) \sim D'$, is equal to $D$. To see this,
\begin{eqnarray}
\Pr_{D'}[\pi = \pi', i=i'] &=& \sum_{S' \subset [n], |S'| = k, S' \ni i'} \Pr_{\mathcal{E}}[S=S'] \cdot \Pr_D[\pi = \pi', i=i' | i \in S']\nonumber\\
&=& \sum_{S' \subset [n], |S'| = k, S' \ni i'}\left( \sum_{i'' \in S'} \frac{\Pr_D[i = i'']} {{n-1 \choose k-1}}\right) \cdot \frac{\Pr_{D}[\pi = \pi', i=i']}{\Pr_D[i \in S']}\nonumber\\
&=& \frac{1}{{n-1 \choose k-1}} \cdot \sum_{S' \subset [n], |S'| = k, S' \ni i'} \Pr_D[\pi = \pi', i=i']\nonumber\\
&=& \Pr_D[\pi=\pi', i=i']\nonumber.
\end{eqnarray}
Recall we wish to lower bound $H_D(\pi(i)  | i)$.  But notice that
$$H_D(\pi(i)  | i) = H_{D'}(\pi(i) | i) \geq H_{D'}(\pi(i) | i, S) = \E_{S' \sim \mathcal{E}}[H_{D'} (\pi(i) | i,S=S')]. 
$$
Hence it suffices to show that for every set $S', |S'| = k$, $H_{D'} (\pi(i) | i,S=S') \geq \log n - \log^{(\ep_2^* - \beta \ep_1^*)}n$ and we do so below.

Fix a subset $S' \subset [n]$, of size $k$, where $k$ also satisfies
\begin{equation}
\label{eq:permassumek}
\delta^{1/4} \cdot n/k \leq \sqrt{2}-1, \quad \delta^{1/4}n \log n/k \leq 1/10, \quad nC/k^2 \leq 1/10, \quad k \leq n/10.
\end{equation}
We remark that for each $\beta > 4$, there is some $n_0$ such that for $n \geq n_0$, such a $k$ satisfying (\ref{eq:permassumek}) always exists. (Recall our assumption above that $\beta > 8$.) 

We will specify the exact value of $k$ below, but for now we note that our argument holds for any $k$ satisfying (\ref{eq:permassumek}). 
By the definition of $D'$, we have that $H_{D'} (\pi(i) | i,S=S') = H_{D_{S'}} (\pi(i) | i)$. We show below that $(\pi(S'),i)$ where $(\pi,i) \sim D_{S'}$ satisfies the preconditions of \Cref{lem:noisyimfn}.
 To show this, we need to choose $\gamma(n,k,\delta) \in [1/n,1/8)$ and $\Gamma(n,k,\delta,C) \leq k$ satisfying the following:
\begin{enumerate}
	\item $H_{D_{S'}}(i | \pi) = H_{D}(i | \pi, i \in S') \geq \log k - \gamma(n,k,\delta)$.
	\item $H_{D_{S'}}(\pi(S')) = H_{D}(\pi(S') | i \in S') \geq k \log n - \Gamma(n,k,\delta,C)$.
\end{enumerate}

The following claim helps with the choice of $\gamma(n,k,\delta)$.

\begin{claim}
  \label{cl:isprime}
  Suppose that $i \in [n]$ is a random variable such that $H(i) \geq \log n - \tau$ with 
  $n\sqrt{\tau}/k \leq \sqrt{2}-1$. Then $H_D(i | i \in S') \geq \log k - \frac{n\sqrt{\tau}}{k}\log\left(\frac{k^2}{n\sqrt{\tau}} \right)$.
\end{claim}
\begin{proof}[Proof of Claim \ref{cl:isprime}]
Let $U_n$ denote the uniform distribution on $[n]$. By Pinsker's inequality we have that, $\Delta(D_i, U_n) \leq \sqrt{\tau/2},$ which in turn implies that $\left| \Pr_{D_i}[i \in S'] - k/n  \right| \leq \sqrt{\tau/2}.$ Let $U_{S'}$ be the uniform distribution over $S'$. We have that
$$\Delta((D_i | i \in S'), U_{S'}) \leq \sqrt{\tau/2}\cdot\frac{1}{k/n - \sqrt{\tau/2}} \leq \frac{n\sqrt{\tau}}{k},$$
since $n\sqrt{\tau}/k \leq \sqrt{2}-1$.
By Theorem~\ref{thm:cover}, we get that,
$$H_D(i | i \in S') \geq \log k - \frac{n\sqrt{\tau}}{k}\log\left(\frac{k}{(n\sqrt{\tau}/k) } \right) = \log k - \frac{n\sqrt{\tau}}{k}\log \left(\frac{k^2}{n\sqrt{\tau}}\right)
$$

\end{proof}

By Markov's inequality, with probability at least $1 - \sqrt{\delta}$ when $\pi' \sim D_\pi$, we have $ H(i | \pi = \pi') \geq \log k - \sqrt{\delta}$. 
For such $\pi'$, by Claim \ref{cl:isprime} applied to the distribution $i | \pi = \pi'$ and $\tau = \sqrt{\delta}$ (note that the condition $n\delta^{1/4}/k = n \sqrt{\tau}/k \leq 1 - 1/\sqrt{2}$ holds by the conditions on $k$), we obtain
$$H_D(i | i\in S', \pi = \pi') \geq \log k - \frac{n\delta^{1/4}}{k}\log\left(\frac{k^2}{\delta^{1/4}n} \right) \geq  \log k - \frac{n\delta^{1/4}}{k}\log\left(\frac{k}{\delta^{1/4}} \right).$$
Hence
$$H_D(i | i\in S', \pi) \geq  (1-\sqrt{\delta})\left(  \log k - \frac{n\delta^{1/4}}{k}\log\left(\frac{k}{\delta^{1/4}} \right) \right) \geq \log k - \gamma(n,k,\delta),$$
where $\gamma(n,k,\delta) = \sqrt{\delta} \log n + \frac{n\delta^{1/4}}{k} \log (n^2)$, where we have used $k \leq n$ and $\delta \geq 1/n$. 

Now we turn to determining $\Gamma(n,k,\delta,C)$ such that $H_{D_{S'}}(\pi(S')) \geq k\log n - \Gamma(n,k,\delta,C)$.
Note that $H(\pi | \One[i \in S']) \geq \log n! - C - 1$. 
Applying Pinsker's inequality to the condition $H(i) \geq H(i|\pi) \geq \log n - \delta$ yields that $\Delta(i, U_n) \leq \sqrt{\delta/2}$, meaning that $\left| k/n - \Pr_D[i \in S'] \right| \leq \sqrt{\delta/2}$. Hence 
\begin{eqnarray}
H_D(\pi | i \in S') &\geq& \frac{\log (n!) \cdot  (k/n - \sqrt{\delta/2}) - C - 1}{k/n + \sqrt{\delta/2}} \nonumber\\
&=& \log (n!) \cdot \frac{1 - \sqrt{\delta/2}n/k}{1 + \sqrt{\delta/2}n/k} - \frac{C+1}{k/n + \sqrt{\delta/2}}\nonumber\\
& \geq & \log (n!) \cdot (1 - \sqrt{2\delta} \cdot n/k) - \frac{C+1}{k/n + \sqrt{\delta/2}}\nonumber\\
& \geq & \log (n!) - n \cdot \left( \sqrt{2\delta} \cdot n \log( n) / k + 2C/k\right)\nonumber,
\end{eqnarray}
where we have used that $n! \leq n^n$. But since $\pi$ is a permutation,
\begin{eqnarray}
H_D(\pi(S') | i \in S') & = & H_D(\pi | i \in S') - H_D(\pi([n] \backslash S') | i \in S', \pi(S'))\nonumber\\
& \geq & \log (n!) - n \cdot \left( \sqrt{2\delta} \cdot n \log( n)  / k + 2C/k\right) - \log( (n-k)! )\nonumber\\
& \geq & k \log(n-k) - n \cdot \left( \sqrt{2\delta} \cdot n \log( n) / k + 2C/k\right) \nonumber\\
& \geq & k \log n  - k \cdot \left( \sqrt{2\delta} \cdot n^2 \log( n)  / k^2 + 2nC/k^2 +  \frac{2k}{n}\right) \nonumber,
\end{eqnarray}
where we have used that $\log(1-x) \geq -2x$ for $0 \leq x \leq 1/2$, as well as $k \leq n/2$. Hence with $\Gamma = \Gamma(n,k,\delta,C) = k \cdot \left( \sqrt{2\delta} \cdot n^2 \log( n) / k^2 + 2nC/k^2 +  \frac{2k}{n}\right) \leq k$ (by our assumption (\ref{eq:permassumek})), we have that $H(\pi(S') | i \in S') \geq k \log( n) - \Gamma$. It follows from Lemma \ref{lem:noisyimfn} that, writing $\gamma = \gamma(n,k,\delta)$,
\begin{equation}
\label{eq:hmds}
H_{D_{S'}}(\pi(i) | i) =H_D(\pi(i) | i, i \in S') \geq \log n - \frac{\Gamma}{k} - 2\sqrt{2\gamma} \cdot \log n.
\end{equation}
Therefore,
\begin{equation}
  \label{eq:hdpii}
H_D(\pi(i) | i) \geq \E_{S \sim \mathcal{E}}[H_{D_{S}}(\pi(i) | m,i)] \geq \log n - \frac{\Gamma}{k} - 2\sqrt{2\gamma} \cdot \log n,
\end{equation}
since the inequality is true for each value $S' \subset [n]$, $|S'| = k$, by (\ref{eq:hmds}).

It is now easily verified that for each $\beta > 8$, for $k = n \cdot \log^{-\beta/8}(n)$, there is some $n_0$, depending only on $\beta$, so that (\ref{eq:permassumek}) is satisfied for $n \geq n_0$. Moreover, for such $k$,
\begin{eqnarray}
  &&\Gamma/k + 2\sqrt{2\gamma} \cdot \log n\nonumber\\
  &\leq& \sqrt{2} \log^{(-\beta/2 + 1 + 2\beta/8)}n + 2 \log^{(-\beta + 2\beta/8)}n + 2 \log^{(-\beta/8)}n + 2 \sqrt 2 \cdot \left(\log^{(-\beta/4 + 3/2)}n + 2 \log^{(-\beta/8 + 3/2 + \beta/16)}n\right)\nonumber\\
  & \leq & 100 \log^{(3/2 - \beta/16)}n\nonumber\\
  & \leq & \log^{(4-\beta/16)}n\nonumber,
\end{eqnarray}
where the last inequality holds for sufficiently large $n$. By (\ref{eq:hdpii}) this implies that for each $\beta > 8$, there is some $n_0$ such that for $n \geq n_0$, $H_D(\pi(i) | i) \geq \log(n) - \log^{(4 - \beta/16)}n$, which completes the proof. 
\end{proof}



Now we are ready to lower bound the entropy $H(\pi(i) | m,i,Z)$, that proves Lemma~\ref{lem:base-case}: Equation~(\ref{eq:lem-part1}), via the following lemma.
\begin{lemma}\label{lem:noisyimz}
	There exists constants $\ep_1^*>0, \ep_2^*$ such that for every $\beta>0$ there exists $n_0$ such that for all $n \geq n_0$ the following holds:
	Let $\delta = 1/\log^{\beta} n$, $C = C' = \delta n$, and $\tilde{\beta} = \ep_1^*\cdot \beta - \ep_2^*$. Suppose $(i,j,\pi,Z)$ are drawn from a distribution $D$, with $Z$ taking on finitely many values, such that the following properties hold:
	\begin{enumerate}
		\item $H(i | \pi, Z) \geq \log(n) - \delta$.
		\item $H(\pi | Z) \geq \log(n!)- C$.
        \end{enumerate}
	Then, for every deterministic function $m = m(\pi, Z)$ with $m \in \{0,1\}^{C'}$, we have
	\begin{align*}
	H(\pi(i) | i,m, Z)& \geq \log(n) - 1/\log^{\tilde{\beta}} n.
	\end{align*}
\end{lemma}	
	
\begin{proof} 
	In Lemma~\ref{lem:noisyimperm} we proved a lower bound on $H(\pi(i) | i)$, given the conditions that $H(i | \pi) \geq \log n - \delta$ and $H(\pi) \geq \log n! - C$.	
	 We would now like to prove a bound on $H(\pi(i) | i, m, Z)$, where $m = m(\pi,Z)$ is a message of length $\leq C'$ and $Z$ is the random variable in the lemma statement. Since $|m| \leq C'$, (1) and (2) in the lemma hypothesis, along with the data processing inequality, imply that, 
	\begin{enumerate}
		\item  $H(i | \pi, m, Z) \geq \log n - \delta$.
		\item $H(\pi | m,Z) \geq \log n! - C - C'$.
	\end{enumerate}

	Let $\gamma = (C+C')/n$, so that $\gamma \leq 2/\log^\beta(n)$. By Markov's inequality (and the facts that $i$ takes on at most $n$ values and $\pi$ takes on at most $n!$ values), we have the following, for every $\ep > 0$:
	\begin{itemize}
		\item With probability at least $1 - \sqrt{\delta}$ over the choice of $(m',z) \sim (m,Z)$, we have that $H(i | \pi, m = m', Z=z) \geq \log(n) - \sqrt{\delta}$.
		\item With probability at least $1 - \sqrt{\gamma}$ over the choice of $(m',z) \sim (m,Z)$, we have that $H(\pi | m = m', Z=z) \geq \log(n!) - n \cdot \sqrt{\gamma}$.
	\end{itemize}
	Let $\alpha = \max\{\delta, \gamma\}$. For sufficiently large $n$ we have that $\sqrt{\alpha} \leq 1/\log^{(\beta/3)}n$. 
	Then by Lemma \ref{lem:noisyimperm}, there is some $n_0$, depending only on $\beta$, such that for all $(m',z)$ belonging to some set of measure at least $1 - 2\sqrt{\alpha}$, for $n \geq n_0$ we have that $H(\pi(i) | i, m=m', Z=z) \geq \log n - \eta$, where $\eta = \log^{\mu_2^* - \beta \mu_1^*}n$, for absolute constants $\mu_1^*, \mu_2^*$. Then there are suitable absolute constants $\ep_1^* \in (0,1), \ep_2^* > 0$ and $n_0'$ (depending only on $\beta$) such that for $n \geq n_0$,
	\begin{eqnarray*}
      H(\pi(i) | i,m,Z)& =& \E_{(m',z) \sim (m,Z)} [H(\pi(i) | i,m=m',Z=z)] \\
                       &\geq& (1 - 2 \sqrt{\alpha}) \cdot (\log(n) - \eta) \\
      &\geq& \log(n) - \log^{(\ep_2^* - \beta \ep_1^*)}n.
	\end{eqnarray*}
\end{proof}

Next we work towards the proof of (\ref{eq:lem-part2}) in Lemma \ref{lem:base-case}. The main difficulty in proving this inequality is to reason about the conditional entropy of the indicator random variable $\One[\pi(i) = j]$, conditioned on the random variable $j$. Roughly speaking, Lemma \ref{lem:jxyz} below allows us to infer a statement such as $H(\One[\pi(i) = j] | j) \geq 1-o(1)$ from an analogous statement of the form $H(\One[\pi(i) = j] | \pi(i)) \geq 1-o(1)$, if $\pi(i), j \in [n]$ satisfy certain regularity conditions. This same argument is needed in the inductive step presented in Lemma \ref{lem:simulation}. In these applications we need to additionally condition all entropies on some random variable $Z$. 
\begin{lemma}
  \label{lem:jxyz}
  There are absolute constants $\ep_1^*>0, \ep_2^*, n_0$ such that the following holds for every $n \geq n_0$:
	Let $X,Y,Z$ be random variables with $X,Y \in [n]$ and $Z$ takes on finitely many values. Let $J = \One[X=Y]$. If there is some constant $\beta > 0$ such that $\delta \leq 1/\log^\beta n$, and
	\begin{enumerate}
		\item $H(X|Z) \geq \log(n) - \delta$.
		\item $H(J|X,Z) \geq 1 - \delta$.
		\item $H(Y | X, Z, J=0) \geq \log(n) - \delta$
	\end{enumerate}
    Then $H(J|Y,Z) \geq 1-\log^{(\ep_2^* - \beta \ep_1^*)}n$.
\end{lemma}

\begin{proof}
We will first prove the above statement assuming that $H(Z) = 0$ and then use Markov's inequality and a union bound to prove the lemma statement for general $Z$. That is, we first prove that if conditions (1), (2), (3) hold without the conditioning on $Z$ then, $H(J | Y) \geq 1 - o(1)$. 

We have that $H(X), H(Y) \leq \log n$ since $X,Y \in [n]$ and $H(J) \leq 1$. Also note that, by Pinsker's inequality,
$$\Pr[J = 0], \Pr[J = 1] \in [1/2 - \sqrt{\delta/2}, 1/2 + \sqrt{\delta/2}].$$

We also have that
\begin{eqnarray}
H(J|Y) &=& H(J) + H(Y|J) - H(Y)\nonumber\\
\label{eq:jyyj}
& \geq & (1-\delta) + H(Y|J) - \log(n) \nonumber \\
& \geq & (1 - \delta) + \Pr[J=0] \cdot H(Y |J=0) + \Pr[J=1] \cdot H(Y|J=1) - \log n \nonumber \\
& \geq & (1-\delta) + (1/2 - \sqrt{\delta/2}) (\log n - \delta + H(Y | J = 1)) - \log n\label{hhh}
\end{eqnarray}

But notice that $H(Y | J = 1) = H(Y | X = Y) = H(X | J = 1)$, so it suffices to bound the latter.

From the lemma hypothesis we get that
$$H(X|J) = H(X) + H(J | X) - H(J) \geq (\log n - \delta) + (1 - \delta) - 1 \geq \log n - 2\delta.$$

On the other hand we have that
\begin{eqnarray}
H(X | J) &=& \Pr[J = 0] \cdot H(X | J = 0) + \Pr[J = 1] \cdot H(X | J = 1) \nonumber\\
& \leq & (1/2 + \sqrt{\delta/2}) \cdot (\log n + H(X | J=1)).
\end{eqnarray}

Combining the upper and lower bounds on $H(X | J)$, we get that
$$H(X | J = 1) \geq \frac{\log(n) - 2\delta}{1/2 + \sqrt{\delta/2}} - \log n \geq \log(n) - 4\delta - \sqrt{8\delta} \log n.
$$
Plugging the above into (\ref{hhh}), we get that,
$$H(J | Y) \geq 1 - \frac{7\delta}{2} - 2\sqrt{\delta} \log n.$$

To get the lower bound while conditioning on $Z$, we use Markov's inequality and a union bound (in the same manner as Lemma~\ref{lem:noisyimz}) to get that
\begin{eqnarray*}
	H(J | Y, Z) & \geq &  (1 -3\sqrt{\delta})\left(1 - \frac{7\sqrt{\delta}}{2} - 2\delta^{1/4} \log n\right) \\
	            & \geq & 1 - 7\sqrt{\delta} - 2\delta^{1/4}\log n \\
	            & \geq & 1 - 9 \delta^{1/4}\log n\\
	            & \geq &    1 - 9\log^{(1 - \beta/4)}n \\
	            & \geq & 1 - \log^{(\ep^*_2 - \beta\ep_1^*)} n,
\end{eqnarray*}
where the final inequality holds for $\ep_1^* = 1/4$, $\ep_2^* = 2$ and $n_0=2^9$ (so that $\log n \geq 9$).
\end{proof}

The proof of Lemma \ref{lem:base-case}: Equation~(\ref{eq:lem-part2}) follows as a consequence of Lemmas~\ref{lem:noisyimz} and \ref{lem:jxyz} above.
\begin{proof}[Proof of Lemma~\ref{lem:base-case}]
    
	We show that there exist $\ep_1^* > 0$ and $\ep_2^*$ such that if ${\beta} \geq (\tilde\beta+\ep_2^*)/\ep_1^*$ (or equivalently, if $\tilde \beta \leq \ep_1^*\cdot \beta - \ep_2^*$) then Equations (\ref{eq:lem-part1}) and (\ref{eq:lem-part2}) of \Cref{lem:base-case} hold for every $n \geq n_0$ where $n_0 = \max\{n_{0,1},n_{0,2}\}$ and $n_{0,1} = n_{0,1}(\beta)$ is as given by \Cref{lem:noisyimz} and $n_{0,2} = n_{0,2}(\beta)$ is the constant given by \Cref{lem:jxyz}.
	For this choice 
   Lemma \ref{lem:noisyimz} already gives us (\ref{eq:lem-part1}), that is, $H(\pi(i) | i,m) \geq \log(n) - \log^{(\mu_2^* - \beta\mu_1^*)}n$ for some absolute constants $\mu_1^*\in (0,1), \mu_2^* > 0$. Note in particular that this implies that for every $\ep_2^* \geq \mu_2^*$ and for every $\ep_1^* \leq \mu_1^*$ we have $H(\pi(i) | i,m) \geq \log(n) - \log^{(\ep_2^* - \beta\ep_1^*)}n$ and we will make such a choice below. 

  We next apply Lemma \ref{lem:jxyz} with $Z^* = (m,i,Z), X = \pi(i), Y = j$, and $J = \One[\pi(i) = j]$, where $Z^*$ refers to the random variable in Lemma~\ref{lem:jxyz} and $Z$ refers to the one in Lemma~\ref{lem:base-case}. We verify that each of the pre-conditions is met. 
\begin{enumerate}
	\item $X, Y \in [n], J \in \{0,1\}$ and $Z^*$ takes finitely many values.
	\item $H(X|Z^*) = H(\pi(i) | i,m,Z) \geq \log(n) - \log^{(\mu_2^* - \mu_1^* \beta)}n$, by (\ref{eq:lem-part1}).
	\item $H(J|X,Z^*) = H(\One[\pi(i) = j] | \pi(i),m,i,Z) \geq H(\One[\pi(i) = j] | \pi, i,Z) \geq 1 - \delta$, by assumption.
	\item $H(Y | X, Z^*, J = 0) = H(j | \pi(i), m,i, Z,\One[\pi(i) = j]) \geq H(j | \pi, i, Z,\One[\pi(i) = j]) \geq 1 - \delta$, by assumption.
\end{enumerate}
Then by Lemma \ref{lem:jxyz}, we have that for $n \geq n_0$,
$$
H(\One[\pi(i) = j] | m,i,j, Z)=H(J|Y,Z^*) \geq 1- \log^{(\nu_2^* - (\mu_2^* - \beta \mu_1^*) \nu_1^*)}n,
$$
where $\nu_1^*, \nu_2^*$ denote the absolute constants of Lemma \ref{lem:jxyz}. Thus again we have that if $\ep_2^* \geq \nu_2^* - \mu_2^* \nu_1^*$ and $\ep_1^* \leq \mu_1^* \nu_1^*$
then we have that $H(\One[\pi(i) = j] | m,i,j, Z) \geq 1-\log^{(\ep_2^* - \beta \ep_1^*)}n$.
Setting  $\ep_1^* = \min\{\mu_1^*,\mu_1^* \nu_1^*\}$ and $\ep_2^* = \max\{\mu_2^*,\nu_2^* - \mu_2^* \nu_1^*\}$ thus ensures that both conditions of the lemma are satisfied.
\end{proof}

\subsection{The Inductive Step: Proof of {\protect Lemma~\ref{lem:simulation}}}
\label{ssec:simulation}


We will prove the inductive step via a simulation argument. That is, we show that if Alice and Bob were able to succeed on $D \in \Dmpv(n,r,\delta,C)$ with non-negligible probability, then they would also succeed on some $\tilde{D} \in \Dmpv(n,r-2,\delta',C')$ by simulating the protocol for $D$ given an instance from $\tilde{D}$. 

Given a distribution $D$ on which Alice and Bob can succeed with non-negligible probability, we consider the distribution $\tilde{D}$ on the resulting ``inner inputs'' (i.e.~the original inputs minus $\pi_1, \pi_r$) after Alice sends a short message to Bob. 
More precisely, the distribution $\tilde{D}$ is the distribution of $(i_1,j_1,\pi_2,\ldots,\pi_{r-1})$ conditioned on Alice's first message $m_1$ and Bob's indices $(i_0,j_0)$, where $(i_1,j_1) = (\pi_1(i_0), \pi_r^{-1}(j_0))$. Moreover, the inputs of $\tilde{D}$ are given to the players as follows: Alice holds $(i_1,j_1, \pi_3, \pi_5, \ldots, \pi_{r-2})$, Bob holds $(\pi_2, \pi_4, \ldots, \pi_{r-1})$, and it is Bob's turn to send the next message. Therefore, 
this corresponds to an instance of an $(r-2)$-Pointer Verification Problem with Alice and Bob's roles flipped. We will show in Lemma~\ref{lem:inductivestep} that $\tilde{D} \in \Dmpv(n,r-2,\delta',C')$, for some $\delta', C'$ not too much larger than $\delta, C$, respectively. Then using the protocol for $D$, we will construct a protocol that succeeds when the inputs are drawn from $\tilde{D}$, with not much loss in the success probability. We will now prove two simple lemmas that will be used to prove Lemma~\ref{lem:inductivestep}.

\begin{lemma}
	\label{lem:noisyimz2}
	There exists $\ep_1^*>0$ and $\ep_2^*$ such that for every $\beta$ there exists $n_0$ such that for all $n \geq n_0$ the following holds:
	Suppose $i,j, \sigma_1, \sigma_2, Z$ are random variables, where $i,j\in [n]$, $\sigma_1,\sigma_2 \in S_n$ and $Z$ takes on finitely many values, satisfying the following conditions:
	\begin{enumerate}
		\item $H(i | \sigma_1, \sigma_2, Z) \geq \log(n) - \delta$, with $\delta \leq 1/\log^\beta n$.
		\item $H(\sigma_1, \sigma_2 | Z) \geq 2 \log(n!) - C$, with $C \leq n/\log^\beta n$.
		\item For each $z$ for which the event $\{Z = z\}$ has positive probability, there is a permutation $f_z : [n] \ra [n]$, such that $f_z(\sigma_1(i)) = \sigma_2(j)$ (which implies that $\sigma_1(i) = f_z^{-1}(\sigma_2(j))$.
	\end{enumerate}
	Suppose further that $m = m(\sigma_1, \sigma_2, Z)$ is a deterministic function and $m \in \{0,1\}^{C'}$, with $C' \leq n/\log^\beta n$.
	Then  $H(\sigma_1(i) | i,j, m, Z) \geq \log n - \log^{(\ep_2^* - \beta \ep_1^*)} n$. 
\end{lemma}

\begin{proof}
	Let us write $Z' = (\sigma_2^{-1} \circ f_Z \circ \sigma_1, Z)$. Then
	\begin{enumerate}
		\item $H(i | \sigma_1, Z') = H(i | \sigma_1, \sigma_2^{-1} \circ f_Z \circ \sigma_1, Z) = H(i | \sigma_1, \sigma_2, Z) \geq \log(n) - \delta$.
		\item $H(\sigma_1 | Z') = H(\sigma_1 | \sigma_2^{-1} \circ f_Z \circ \sigma_1, Z)  \geq \log(n!) - C$, where the last inequality follows from the following:
		\begin{eqnarray}
		2\log(n!) - C & \leq & H(\sigma_1, \sigma_2 | Z)\nonumber\\
		&=& H(\sigma_2^{-1} \circ f_Z \circ \sigma_1, \sigma_2 | Z)\nonumber\\
		&=& H(\sigma_2^{-1} \circ f_Z \circ \sigma_1 | Z) + H(\sigma_1 | \sigma_2^{-1} \circ f_Z \circ \sigma_1, Z)\nonumber\\
		& \leq & \log(n!) + H(\sigma_1 | \sigma_2^{-1} \circ f_Z \circ \sigma_1, Z)\nonumber.
		\end{eqnarray}
	\end{enumerate}
	Then by Lemma~\ref{lem:noisyimz}, $H(\sigma_1(i) | i,m,Z') = H(\sigma_1(i) | i,m,\sigma_2^{-1} \circ f_Z \circ \sigma_1, Z) \geq \log n - \log^{(\ep_2^* - \beta \ep_1^*)} n$, for absolute constants $\ep_1^*, \ep_2^*$ and for $n$ sufficiently large as a function of $\beta$. But since $j = \sigma_2^{-1} \circ f_Z \circ \sigma_1(i)$, we obtain that
	$$
	H(\sigma_1(i) | i,j,m,\sigma_2^{-1} \circ f_Z \circ \sigma_1, Z) \geq \log n - \log^{\ep_2^* - \beta \ep_1^*} n.
	$$
	Then the desired result follows since conditioning decreases entropy.
\end{proof}

\begin{lemma}
	\label{lem:abcindep}
	Suppose $A,B,C$ are random variables with finite ranges such that $A \indep B\ |\ C$. Let $\Omega_A$ denote the domain of $A$, and $f : \Omega_A \ra \{0,1\}^*$ be a function. It follows that
	$$
	A \indep B \ \ | \ \ \{C, f(A)\}.
	$$
\end{lemma}
\begin{proof}
	Pick any $x \in \{0,1\}^*$, $a \in \Omega_A, b \in \Omega_B, c \in \Omega_C$. We have that
	\begin{eqnarray}
	&& \Pr[A=a, B=b | C=c, f(A) = x]\nonumber\\
	\label{eq:abcfx}
	&=& \frac{\Pr[A=a, B=b, f(A) = x | C=c]}{\Pr[ f(A) = x|C=c]}.
	\end{eqnarray}
	If $f(a) \neq x$, then the above is 0, and also
	$$
	\Pr[A=a | C=c,f(A) = x] \cdot \Pr[B=b | C=c, f(A) = x]= 0
	$$
	as well. If $f(a) = x$, then (\ref{eq:abcfx}) is equal to
	\begin{eqnarray}
	\frac{\Pr[A=a, B=b | C=c]}{\Pr[ f(A) = x|C=c]} &=& \frac{\Pr[A=a | C=c]}{\Pr[ f(A) = x|C=c]} \cdot \Pr[B=b | C=c]\nonumber\\
	&=& \frac{\Pr[A=a, f(A) = x | C=c]}{\Pr[ f(A) = x|C=c]} \cdot \Pr[B=b |f(A) = x, C=c]\nonumber\\
	&=& \Pr[A=a | f(A) = x, C=c] \cdot \Pr[B=b | C=c, f(A) = x] \nonumber,
	\end{eqnarray}
	where the second-to-last inequality follows since
	$$
	\Pr[B=b | C=c] = \Pr[B=b | f(A) = x, C=c],
	$$
	as $B$ is conditionally independent of $A$ given $C$.
\end{proof}


Given a distribution $D \in \Dmpv(n,r,\delta,C)$ and a deterministic function $m = m(\pi_A)$ we define a distribution $\tilde{D}^+$ on the $r-2$ permutation pointer verification problem with some auxiliary randomness $Z$ as follows:
To generate a sample $(\pi_2,\ldots,\pi_{r-2},i_1,j_1; Y)$ according to $\tilde{D}^+$ we first sample  $(\pi_1,\ldots,\pi_r,i_0,j_0)\sim D$ and let $i_1 = \pi_1(i_0)$, $j_1 = \pi_r^{-1}(j_0)$ and 
 $Y =(m_1(\pi_A),i_0,j_0)$. 
 
$\tilde{D}^+$ as defined above is a candidate ``noisy-on-average' (i.e., noisy when averaged over $Y$ --- see last paragraph of \Cref{def:dmixdef}) distribution on $r-2$ permutations, and the lemma below asserts that this is indeed the case for slightly larger values of $\delta$ and $C$ provided $|m|$ is small. 
Recall that $\pi_A = (\pi_1, \pi_3, \ldots, \pi_r), \pi_B = (\pi_2, \pi_4, \ldots, \pi_{r-1})$.

\begin{lemma}
\label{lem:inductivestep}
There exist constants $\ep^*_1 > 0, \ep^*_2$ such that for every odd $r\geq 3$ and $\beta > 0$ there exists $n_0$ such that for every $n \geq n_0$ the following holds:
Suppose $D \in \Dmpv(n,r,\delta,C)$, for some $\delta \leq 1/\log^\beta n$ and $C\leq n/\log^\beta n$. Also suppose that $C' \leq n/\log^\beta n$, and  that $m = m(\pi_A)$ is a deterministic function of $\pi_A$ such that $|m| \leq C'$. Then for $\delta' = \log^{(\ep_2^* - \ep_1^* \cdot \beta)} n$ we have $\tilde{D}^+ \in \Dmpvp(n,r-2,\delta',\delta'n)$.

\end{lemma}


\begin{proof}[Proof of Lemma~\ref{lem:inductivestep}]
  We need to verify statements (1) -- (5) of \Cref{def:dmixdef} in order to show that $\tilde{D}^+ \in \Dmpvp(n,r-2, \delta', \delta'n)$, for an appropriate choice of $\ep_1^*, \ep_2^*$ and for sufficiently large $n$ (depending only on $\beta$). We will show that statement (5) (which does not depend on $\delta'$) holds for all $n \in \mathbb{N}$. To verify statements (1) -- (4), we will show that for each of these statements, there are some absolute constants $\hat{\ep}_1^*, \hat{\ep}_2^*$ and some $\hat n_0$ (depending only on $\beta$) such that for $n \geq \hat n_0$, the statement holds with $\delta' = \log^{(\hat{\ep}_2^* - \hat{\ep}_1^* \beta)}n$. The proof of the lemma will follow by choosing $\ep_2^*$ to be the maximum of the individual $\hat{\ep}_2^*$, ${\ep}_1^*$ to be the minimum of the individual $\hat{\ep}_1^*$, and $n_0$ to be the maximum of the individual $\hat{n}_0$.

  We now proceed to verify each of the statements (1) -- (5). We remark that the values of $\hat{\ep}_1^*, \hat{\ep}_2^*, \hat{n}_0$ may change from line to line.
  \begin{enumerate}
  \item We first verify that $H(i_1 | i_0, j_0, \pi_2, \ldots, \pi_{r-1}, m) \geq \log(n) -\delta'$. Since conditioning can only reduce entropy, it suffices to find a lower bound on $H(i_1 | \One[\pi_1^r(i_0) = j_0], i_0, j_0, \pi_2, \ldots, \pi_{r-1}, m)$, and in particular, it suffices to find a lower bound on $H(i_1 | \pi_1^r(i_0) \neq j_0, i_0, j_0, \pi_2, \ldots, \pi_{r-1}, m)$ and on $H(i_1 | \pi_1^r(i_0) = j_0, i_0, j_0, \pi_2, \ldots, \pi_{r-1}, m)$. 

We first bound the former. Consider the distribution of $i_0, j_0, \pi_1, \pi_2, \ldots, \pi_r$ conditioned on the event $\pi_1^r(i_0) \neq j_0$, and let $Z = (\pi_2, \pi_3, \ldots, \pi_{r-1}, \pi_r)$. We will now use Lemma \ref{lem:noisyimz} with $i = i_0, \pi = \pi_1$, and with the distribution being $D$ conditioned on $\pi_1^r(i_0) \neq j_0$. To apply this lemma, we first verify its preconditions:
\begin{enumerate}
\item $H(i_0 | \pi_1, Z, \pi_1^r(i_0) \neq j_0) \geq \log(n) - 5\delta$ as long as $n$ is large enough so that $\delta \leq 1/50$. To see this, conditions (1a) and (3a) of the distribution $D \in \dmpv(n,r,\delta, C)$ (recall \Cref{def:dmixdef}) imply that
$$
H((i_0, \One[\pi_1^r(i_0) = j_0]) | \pi_1, \pi_2, \ldots, \pi_r) \geq 1 + \log(n) - 2\delta,
$$
meaning that 
\begin{eqnarray*}
&& \lefteqn{H(i_0 | \One[\pi_1^r(i_0) = j_0], \pi_1, \pi_2, \ldots, \pi_r)} \\
&=& \Pr[\pi_1^r(i_0) = j_0] \cdot H(i_0 | \pi_1^r(i_0) = j_0, \pi_1, \ldots, \pi_r) \\
& & + \Pr[\pi_1^r(i_0) \neq j_0] \cdot H(i_0 | \pi_1^r(i_0) \neq j_0, \pi_1, \ldots, \pi_r) \\
&\geq &\log(n) - 2\delta.
\end{eqnarray*}
By Pinsker's inequality and condition (3a) of $D$ we have that $\left| \Pr[\pi_1^r(i_0) = j_0] - 1/2 \right| \leq \sqrt{\delta/2}$, so for sufficiently small $\delta$ (in particular, such that $\sqrt{\delta/2} \leq 1/10$), it follows that
\begin{equation}
\label{eq:i1not1}
\min\left\{H(i_0 | \pi_1^r(i_0) = j_0, \pi_1, \ldots, \pi_r) , H(i_0 | \pi_1^r(i_0) \neq j_0, \pi_1, \ldots, \pi_r) \right\} \geq \log(n) - 5\delta.
\end{equation}

\item $H(\pi_1 | Z, \pi_1^r(i_0) \neq j_0) \geq \log(n!) - 3C - 3\delta$ as long as $n$ is large enough so that $\delta \leq 1/18$. The proof is similar to (a) above. In particular, condition (2) of the distribution $D$ implies that
$$
H(\pi_1 | \pi_2, \pi_3, \ldots, \pi_r) \geq \log(n!) - C.
$$
Since conditioning can only reduce entropy, condition (3a) of the distribution $D$ implies that
$$
H((\pi_1, \One[\pi_1^r(i_0) = j_0]) | \pi_2, \ldots, \pi_r) \geq 1 + \log(n!) - C - \delta,
$$
meaning that
\begin{eqnarray}
&& H(\pi_1 | \One[\pi_1^r(i_0) = j_0], \pi_2, \ldots, \pi_r)\nonumber\\
&=& \Pr[\pi_1^r(i_0) = j_0] \cdot H(\pi_1 | \pi_1^r(i_0) = j_0, \pi_2, \ldots, \pi_r) + \Pr[\pi_1^r(i_0) \neq j_0] \cdot H(\pi_1 | \pi_1^r(i_0) \neq j_0, \pi_2, \ldots, \pi_r)\nonumber\\
& \geq & \log(n!) - C - \delta.\nonumber
\end{eqnarray}
By Pinsker's inequality and condition (3a) of $D$ we have that $\left| \Pr[\pi_1^r(i_0) = j_0] - 1/2 \right| \leq \sqrt{\delta/2}$, so for sufficiently small $\delta$ (in particular, such that $\sqrt{\delta/2} \leq 1/6$), it follows that
\begin{equation}
\label{eq:c1not1}
\min \{ H(\pi_1 | \pi_1^r(i_0) = j_0, \pi_2, \ldots, \pi_r), H(\pi_1 | \pi_1^r(i_0) \neq j_0, \pi_2, \ldots, \pi_r) \} \geq \log(n!) - 3C - 3\delta.
\end{equation}

\end{enumerate}
Note also that indeed $m$ is a deterministic function of $(\pi, Z) = (\pi_1, \pi_2, \ldots, \pi_r)$. Therefore, by Lemma~\ref{lem:noisyimz}, we obtain that there are absolute constants $\hat{\ep}_1^*, \hat{\ep}_2^*$, such that for some $\hat{n}_0$ depending only on $\beta$, if $n \geq \hat n_0$,
$$
H(i_1 | i_0, \pi_2, \ldots, \pi_r, m, \pi_1^r(i_0) \neq j_0) \geq \log(n) - \log^{(\hat{\ep}_2^* - \beta \hat{\ep}_1^*)} n.
$$
Condition (4a) of the distribution $D$ implies that
$$
H(j_0 | i_0, \pi_1, \pi_2, \ldots, \pi_r, \pi_1^r(i_0) \neq j_0) = H(j_0 | i_0, i_1, m, \pi_1, \pi_2, \ldots, \pi_r, \pi_1^r(i_0) \neq j_0) \geq \log(n) - \delta.
$$
Since conditioning can only reduce entropy we have from the two above equations that
$$
H((i_1, j_0) | i_0, m, \pi_2, \pi_3, \ldots, \pi_r, \pi_1^r(i_0) \neq j_0) \geq 2\log(n) - \log^{(\hat{\ep}_2^* - \beta \hat{\ep}_1^*)} n - \delta,
$$
so
$$
H(i_1 | i_0, j_0, m, \pi_2, \pi_3, \ldots, \pi_r, \pi_1^r(i_0) \neq j_0) \geq \log(n) -  \log^{(\hat{\ep}_2^* - \beta \hat{\ep}_1^*)} n - \delta,
$$
as desired.

Next we lower bound $H(i_1 | \pi_1^r(i_0) = j_0, i_0, j_0, \pi_2, \ldots, \pi_{r-1}, m)$ using Lemma \ref{lem:noisyimz2} with $Z = (\pi_2, \pi_3, \ldots, \pi_{r-1})$, $\sigma_1 = \pi_1, \sigma_2 = \pi_r^{-1}$, and with the distribution being $D$ conditioned on $\pi_1^r(i_0) = j_0$. We first verify that the lemma's preconditions hold:
\begin{enumerate}
\item The fact that $H(i_0 | \pi_1, \pi_r, Z, \pi_1^r(i_0) = j_0) \geq \log(n) - 5\delta$ for $\delta \leq 1/50$ was proven in (\ref{eq:i1not1}).
\item To verify that $H(\pi_1, \pi_r |Z, \pi_1^r(i_0) = j_0) \geq \log(n) - 3C - 3\delta$ for $\delta \leq 1/18$, we may exactly mirror the proof of (\ref{eq:c1not1}) except for replacing $\pi_1$ with $(\pi_1, \pi_r)$ (and removing $\pi_r$ from the random variables being conditioned on). We omit the details.
  \item Since we are conditioning on $\pi_1^r(i_0) = j_0$, we have that $\pi_r^{-1}(j_0) = \pi_{r-1} (\cdots \pi_2(\pi_1(i_0)))$, which means that we may take $f_Z = \pi_{r-1} \circ \cdots \circ \pi_2$. 
\end{enumerate}
Note also that indeed $m$ is a deterministic function of $(\sigma_1, \sigma_2, Z) = (\pi_1, \pi_2, \ldots, \pi_{r-1}, \pi_r^{-1})$. Then by Lemma \ref{lem:noisyimz2}, it follows that for some absolute constants $\hat{\ep}_1^*, \hat{\ep}_2^*$, there is some $\hat n_0$ (depending only on $\beta$) such that for $n \geq \hat n_0$, $H(i_1 | i_0, j_0, m, Z, \pi_1^r(i_0) = j_0) \geq \log n - \log^{(\hat{\ep}_2^* - \beta \hat{\ep}_1^*)}n$.

By the previous discussion, it then follows that for some absolute constants $\hat{\ep}_1^*, \hat{\ep}_2^*$, there is some $\hat n_0$ (depending only on $\beta$) such that for $n \geq \hat n_0$, $H(i_1 | i_0, j_0, m, \pi_2, \ldots, \pi_{r-1}) \geq \log n - \log^{(\hat{\ep}_2^* - \beta \hat{\ep}_1^*)}n$. 

In an identical manner, using conditions (1b), (2), (3b), (4b) of the distribution $D \in \dmpv(n,r,\delta,C)$, we obtain that for the same $\hat{\ep}_1^*, \hat{\ep}_2^*, \hat n_0$, if $n \geq n_0$ then $H(j_1 | i_0, j_0, m, \pi_2, \ldots, \pi_r) \geq \log(n) - \log^{(\hat{\ep}_2^* - \beta \hat{\ep}_1^*)}n$.

\item To prove statement (2) we claim that $H(\pi_2, \ldots, \pi_{r-1} | m, i_0, j_0) \geq (r-2) \log(n!) - C - C' - 2 \log(n)$; to see this note that
\begin{eqnarray}
  H(\pi_2, \ldots, \pi_{r-1} | m, i_0, j_0) &=& H(\pi_2, \ldots, \pi_{r-1}) + H(m, i_0, j_0 | \pi_2, \ldots, \pi_{r-1}) - H(m, i_0, j_0) \nonumber\\
  & \geq & H(\pi_2, \ldots, \pi_{r-1}) - H(m, i_0, j_0)\nonumber\\
&\geq & (r-2) \log(n!) - C - C' - 2\log(n)\nonumber,
\end{eqnarray}
since $|m| \leq C'$. It readily follows that there exist absolute constants $\hat{\ep}_1^*, \hat{\ep}_2^*$ and some $\hat n_0$ (depending only on $\beta$) such that for $n \geq \hat n_0$, $H(\pi_2, \ldots, \pi_{r-1} | m, i_0, j_0) \geq (r-2)\log(n!) - n\log^{(\hat{\ep}_2^* - \hat{\ep}_1^* \beta)}n$.


\item We will next prove that 3(b) holds by applying Lemma~\ref{lem:base-case}, with $i = i_0, j = j_{r-1}, \pi = \pi_1, Z = (\pi_2,\ldots,\pi_r)$ (recall that $j_{r-1} = \pi_2^{-1} \circ \cdots \circ \pi_r^{-1}(j_0)$). We will first verify that the preconditions of Lemma~\ref{lem:base-case} hold:
\begin{enumerate}
	\item $H(i | \pi, Z) = H(i_0 | \pi_1, \pi_2,\ldots,\pi_r)  \geq \log(n) - \delta$, by condition (1) of the distribution $D$.
	\item $H(\pi | Z) = H(\pi_1 | \pi_2,\ldots,\pi_r) \geq \log n! - C$, by condition (2) of the distribution $D$.
	\item $\begin{aligned}[t]
		H(\One[\pi(i) = j] | \pi, i, Z) &= H(\One[\pi_1(i_0) = j_{r-1}] | \pi_1, i_0, \pi_2,\ldots,\pi_r ) \\ &= H(\One[\pi_1^r(i_0) = j_0] | i_0,\pi_1, \pi_2,\ldots,\pi_r ) \\ &\geq 1-\delta,
      \end{aligned}$
      
      by condition (3) of the distribution $D$.
	
	\item $\begin{aligned}[t]
	H(j | \pi, i, \pi(i) \neq j, Z) &= H(j_{r-1} | \pi_1, i_0, \pi_1(i_0) \neq j_{r-1}, \pi_2,\ldots,\pi_r) \\ 
	&= H(j_{0} | i_0, \pi_1, \ldots,\pi_r, \pi_1^r(i_0) \neq j_{0}) \\ &\geq \log(n) - \delta,
  \end{aligned}$
  
  by condition (4) of the distribution $D$,
\end{enumerate}
where by assumption, there exists $\beta > 0$ such that $\delta, C, C'$ are such that $\max\{ \delta, C/n, C'/n\} \leq 1/\log^\beta(n)$. Moreover, $m$ is a deterministic function of $(\pi, Z) = (\pi_1, \ldots, \pi_r)$. 
Therefore by Lemma~\ref{lem:base-case} we get that, for some absolute constants $\hat{\ep}_1^*, \hat{\ep}_2^*$, and for some $\hat n_0$ (depending only on $\beta$), 
\begin{align}
	H(\One[\pi_1(i_0) = j_{r-1}] | \pi_2,\ldots,\pi_r, m, i_0, j_{r-1}) 
	&= H(\One[\pi(i) = j] | m,i,j, Z) \\ &\geq 1-\log^{(\hat{\ep}_2^* - \beta \hat{\ep}_1^*)}n,
\end{align}
Since $j_0 = \pi_r \circ \cdots \circ \pi_2(j_{r-1})$ and $j_1 =  \pi_{r-1} \circ \cdots \circ \pi_2(j_{r-1})$, by the data processing inequality, we get that for $n \geq \hat n_0$,
\begin{align*}
H(\One[\pi_2^{r-2}(i_1) = j_{1}] | j_1, \pi_2,\ldots,\pi_{r-1}, m, i_0, j_{0}) &\geq H(\One[\pi_1(i_0) = j_{r-1}] | \pi_2,\ldots,\pi_r, m, i_0, j_{r-1}) \\
&\geq   1-\log^{(\hat{\ep}_2^* - \beta \hat{\ep}_1^*)}n.
\end{align*}

The proof of 3(a) (with the same $\hat{\ep}_1^*, \hat{\ep}_2^*, \hat n_0$) follows in a symmetric manner.

\item Next we lower bound $H(j_1 | i_1, \pi_2, \ldots, \pi_{r-1}, \pi_1^r(i_0) \neq j_0, m, i_0, j_0)$. We apply Lemma~\ref{lem:noisyimz} with $Z = (i_0, \pi_1, \pi_2, \ldots, \pi_{r-1})$, $i = j_0$, $\pi = \pi_r^{-1}$, with the distribution given by $D$ conditioned on $\pi_1^r(i_0) \neq j_0$. We first verify that the preconditions are met:
\begin{enumerate}
\item $H(j_0 | \pi_r, Z, \pi_1^r(i_0) \neq j_0) = H(j_0 | i_0, \pi_1, \pi_2, \ldots, \pi_r, \pi_1^r(i_0) \neq j_0) \geq \log(n) - \delta$, by condition (4a) of the distribution $D$.
\item As long as $n$ is large enough so that $\delta \leq 1/18$,
  $$
  H(\pi_r | Z, \pi_1^r(i_0) \neq j_0) = H(\pi_r | i_0, \pi_1, \pi_2, \ldots, \pi_{r-1}, \pi_1^r(i_0) \neq j_0) \geq \log(n!) - 3C - 3\delta - \log n,
  $$
  by an argument identical to that used to prove (\ref{eq:c1not1}), as well as the fact that $i_0 \in [n]$, meaning that its entropy is at most $\log n$.
\end{enumerate}
Moreover, $m$ is a deterministic function of $(\pi, Z) = (\pi_1, \pi_2, \ldots, \pi_r, i_0)$. Then by Lemma~\ref{lem:noisyimz}, it follows that there are absolute constants $\hat{\ep}_1^*, \hat{\ep}_2^*$ and some $\hat n_0$ (depending only on $\beta$) such that for $n \geq \hat n_0$,
\begin{eqnarray*}
  \lefteqn{H(j_1 | i_0, \pi_1, \pi_2, \ldots, \pi_{r-1}, j_0, m, \pi_1^r(i_0) \neq j_0)}\\ &=& H(j_1 | i_1, \pi_1, \pi_2, \ldots, \pi_{r-1}, m, i_0, j_0, \pi_1^r(i_0) \neq j_0) \\
                                                                               &\geq& \log(n) - \log^{(\hat{\ep}_2^* - \beta \hat{\ep}_1^*)}n,
\end{eqnarray*}
which proves the desired statement since conditioning can only reduce entropy. Similarly, conditions (2), (3b), (4b) of $D$ imply in a symmetric manner that  for $n \geq \hat n_0$, $H(i_1 | j_1, \pi_2, \pi_3, \ldots, \pi_{r-1}, \pi_1^r(i_0) \neq j_0, m, i_0, j_0) \geq \log(n) - \log^{(\hat{\ep}_2^* - \beta \hat{\ep}_1^*)} n$.

\item To prove statement (5), first take $t$ odd, and let $X = \pi_A \cap (\pi_1, \pi_2, \ldots, \pi_t, \pi_{r-t+1}, \ldots, \pi_{r-1}, \pi_r)$, $Y  = \pi_B$, and note that condition (5) of the distribution $D$ states that conditioned on:
{\small$$
E := \{ (i_0, \ldots, i_t) = (i_0', \ldots, i_t'), (j_0, \ldots, j_t) = (j_0', \ldots, j_t'), (\pi_{t+2}, \pi_{t+4}, \ldots, \pi_{r-t-1}) = (\pi_{t+2}', \pi_{t+4}', \ldots, \pi_{r-t-1}') \},
$$}
we have that $X$ is independent of $Y$. Note that conditioned on $E$, $m = m(\pi_1, \pi_3, \ldots, \pi_r)$ is a deterministic function of $(\pi_1, \pi_3, \ldots, \pi_t, \pi_{r-t+1}, \pi_{r-t+3}, \ldots, \pi_r) = X$. It follows by Lemma \ref{lem:abcindep} that $X \indep Y | E, m=m'$, which implies that
$$
\pi_A \cap (\pi_2, \ldots, \pi_t, \pi_{r-t+1}, \ldots, \pi_{r-1}) \indep \pi_B | E, m=m'.
$$

Next take $t$ even, take $X = \pi_A$, $Y = \pi_B \cap (\pi_1, \pi_2, \ldots, \pi_t, \pi_{r-t+1}, \ldots, \pi_r)$, and conditioned on:
{\small$$
E := \{ (i_0, \ldots, i_t) = (i_0', \ldots, i_t'), (j_0, \ldots, j_t) = (j_0', \ldots, j_t'), (\pi_{t+2}, \pi_{t+4}, \ldots, \pi_{r-t-1}) = (\pi_{t+2}', \pi_{t+4}', \ldots, \pi_{r-t-1}') \},
$$}$X$ is independent of $Y$. Note that conditioned on $E$, $m = m(\pi_1, \pi_3, \ldots, \pi_r)$ is a deterministic function of $X$. It follows by Lemma \ref{lem:abcindep} that $X \indep Y  | E, m=m'$, which implies that
$$
\pi_B \cap (\pi_2, \ldots, \pi_t, \pi_{r-t+1}, \ldots, \pi_{r-1}) \indep \pi_A | E,m=m'.
$$
\end{enumerate}
\end{proof}


\Cref{lem:inductivestep} establishes that the ``inner input'' (after removing the $\pi_1$ and $\pi_r$ and pushing pointers inwards) is from a noisy distribution (according to \Cref{def:dmixdef}) when averaged over the auxiliary variable $Y$. Intuitively this should imply that the pointer verification problem remains as hard (with one fewer round of communication), but this needs to be shown formally. In particular, Alice and Bob do have additional information such as $\pi_1,\pi_r,i_0,j_0,m$ and all of this might help determine  $\One[\pi_2^{r-1}(i_1) = j_1]$. 

In Lemma \ref{lem:simulation} we formalize this intuition by creating an $(r-1)/2$ round protocol for a noisy distribution $\tilde{D}$ on $r-2$ permutations, using an $(r+1)/2$ round protocol for a related noisy distribution $D$ solving the pointer verification problem on $r$ permutations. This argument makes use of Property (5) of \Cref{def:dmixdef}, which we have not really used yet (except to argue that it holds inductively).
%

\begin{proof}[Proof of Lemma~\ref{lem:simulation}]
  Let $\ep_1^*, \ep_2^*$ be the absolute constants from Lemma \ref{lem:inductivestep}. We will show that we can take $\beta_1 = \max\left\{ \beta_2, \frac{2\beta_2 + \ep_2^*}{\ep_1^*} \right\}$.
  
Let $D \in \Dmpv(n,r,1/\log^{\beta_1}n, n/\log^{\beta_1}n)$ and let $\Pi$ be a protocol for $D$ with communication at most $n/\log^{\beta_1}n$. For sufficiently large $n$, we will give a distribution $\tilde D \in \Dmpv(n,r-2,1/\log^{\beta_2}n, n/\log^{\beta_2}n)$, and will construct a protocol $\tilde \Pi$ for $\tilde D$, which uses no more communication than $\Pi$, and crucially uses one less round of communication than $\Pi$.


\paragraph{Definition of $\tilde{D}$.}
We denote the messages in each round of $\Pi$ by $m_1, \ldots, m_{(r+1)/2}$. Recall that Alice sends $m_1 = m_1(\pi_1,\pi_3,\ldots,\pi_{r})$, Bob sends $m_2 = m_2(m_1,i_0,j_0,\pi_2,\pi_4,\ldots,\pi_{r-1})$, Alice sends $m_3 = m_3(m_1,m_2,\pi_1,\pi_3,\ldots)$, and so on. Let $(m_1',i_0',j_0')$ be a fixed instantiation of the random variables $(m_1,i_0,j_0)$. Given the distribution $D$ on $(i_0, i_1, j_0, j_1\pi_1,\ldots,\pi_r)$, consider the conditional distribution ${D_{m_1',i_0',j_0'} := D | (m_1 = m_1', i_0 = i_0', j_0 = j_0')}$ on $(i_1, j_1, \pi_1,\ldots,\pi_r)$. Furthermore, let $\tilde{D}_{m_1',i_0',j_0'}$ denote the marginal distribution of $D_{m_1',i_0',j_0'}$ on the inner inputs, that is, $(i_1, j_1, \pi_2,\ldots,\pi_{r-1})$. One can interpret $\tilde{D}_{m_1',i_0',j_0'}$, as an $(r-2)$-PV problem, and we will show how, for {\it each} tuple $(m_1', i_0', j_0')$, Alice and Bob can simulate the protocol $\Pi$, given an instance from $\tilde{D}_{m_1',i_0',j_0'}$. 
We will then show how it follows that for {\it some} tuple $(m_1', i_0', j_0')$ this simulation will have success probability at least $1/2 + \ep_2$ and moreover for this tuple $\tilde{D}_{m_1', i_0', j_0'} \in \Dmpv(n, r-2, 1/\log^{\beta_2}n, n/\log^{\beta_2}n)$. 

\paragraph{The protocol $\tilde{\Pi}$.}
Consider any tuple $(m_1', i_0', j_0')$, and an instance of $(r-2)$-PV drawn from $\tilde{D} = \tilde{D}_{m_1',i_0',j_0'}$. We use the symbol $\tilde{}$ for the random variables drawn from $\tilde{D}$. We label the $r-2$ permutations drawn from $\tilde{D}$ as $\tilde{\pi}_2,\ldots, \tilde{\pi}_{r-1}$ (instead of $\pi_1,\ldots,\pi_{r-2}$), the initial indices as $(\tilde{i}_1, \tilde{j}_1)$ (instead of $(i_0,j_0)$). The roles of Alice and Bob are also flipped, in that Bob receives $\tilde \pi_2,\tilde \pi_4, \ldots, \tilde \pi_{r-1}$, and Alice receives $\tilde i_1, \tilde j_1,\tilde \pi_3, \ldots, \tilde\pi_{r-2}$. The goal is to determine whether $\tilde\pi_2^{r-1}(\tilde i_1) = \tilde j_1$. The protocol $\tilde \Pi$ for $\tilde{D}$ is constructed as follows:

\begin{enumerate}
	\item Bob sends the first message $\tilde m_2 := m_2(m_1', i_0', j_0', \tilde \pi_2, \ldots, \tilde \pi_{r-1})$. Recall that $m_2$ was the second message of the protocol $\Pi$. 
	\item Alice then draws $(\tilde \pi_1, \tilde \pi_r)$ from its marginal in $D_{m_1',i_0',j_0'}$, conditioned on the event $\{i_1 = \tilde i_1, j_1 = \tilde j_1, \pi_3 = \tilde \pi_3, \pi_5 = \tilde \pi_5, \ldots, \pi_{r-2} = \tilde \pi_{r-2}\}$, using private randomness. That is,
	\begin{equation}
	\label{eq:tildepi1r}
	(\tilde \pi_1, \tilde \pi_r) \sim [(\pi_1, \pi_r)_{D_{m_1',i_0',j_0'}} | \{ i_1 = \tilde i_1, j_1 = \tilde j_1, \pi_3 = \tilde \pi_3, \pi_5 = \tilde \pi_5, \ldots, \pi_{r-2} = \tilde \pi_{r-2}\}].
	\end{equation}
	\item After receiving $\tilde m_2$ from Bob, Alice then sends $\tilde{m_3} := m_3(m_1',\tilde{m_2},\tilde{\pi_1},\tilde{\pi_3},\ldots,\tilde{\pi_r})$.
	Starting with Alice's $\tilde m_3$, Alice and Bob just simulate the remaining $(r+1)/2 - 2$ rounds of the protocol $\Pi$ (as well as an additional output bit at the end), where Alice takes as her input $\tilde \pi_1, \tilde \pi_3, \ldots, \tilde \pi_{r-2}, \tilde \pi_r$ and Bob takes as his input $i_0', j_0', \tilde \pi_2, \ldots, \tilde \pi_{r-1}$. 
  \end{enumerate}

Since the messages of $\tilde \Pi$ are given by $m_2, m_3, \ldots, m_{(r+1)/2}$ for appropriate inputs of $\Pi$, $\tilde \Pi$ has $(r-1)/2$ rounds, and the communication complexity of $\tilde \Pi$ is no greater than the communication complexity of $\Pi$, namely $n/\log^{\beta_1}n$.

\paragraph{Success Probability.}
Now we will prove that for each tuple $(i_0', j_0', m_1')$, the success probability of $\tilde \Pi$ when inputs are drawn from $\tilde{D}_{i_0', j_0', m_1'}$ is equal to the success probability of $\Pi$ on the distribution $D$ conditioned on $\{ m_1 = m_1', i_0 = i_0', j_0 = j_0'\}$. This will ultimately allow us to choose an appropriate tuple $(i_0', j_0', m_1')$ for which $\tilde \Pi$ achieves success probability at least $1/2 + \ep_2$ on $\tilde{D}_{i_0', j_0', m_1}$.

  Notice that the protocol $\tilde \Pi$ induces a distribution on $(\tilde \pi_1,\ldots,\tilde\pi_r,\tilde i_1,\tilde j_1)$, which we will denote by $\tilde{D}_\Pi$, where $(\tilde i_1,\tilde j_1,\tilde \pi_2,\ldots,\tilde\pi_{r-1})$ is drawn from $\tilde{D}_{m_1',i_0',j_0'}$ and Alice draws $(\tilde\pi_1,\tilde\pi_r)$ from the conditional distribution specified in step (2) above, using private randomness. 

We claim that the distribution of $(\tilde i_1, \tilde j_1, \tilde \pi_1, \ldots, \tilde \pi_r)$ under $\tilde D_\Pi$ is the same as the distribution of $(i_1, j_1, \pi_1, \ldots, \pi_r)$ under $D_{m_1',i_0',j_0'}$.
One can think of drawing $(i_1, j_1, \pi_1, \ldots, \pi_r)$ from $D_{m_1',i_0',j_0'}$ as first drawing $(i_1,j_1,\pi_2,\ldots,\pi_{r-1})$ from its marginal distribution $\tilde{D}_{m_1', i_0', j_0'}$ and then drawing $(\pi_1,\pi_r)$ from $D_{m_1',i_0',j_0'} | \{(i_1, j_1, \pi_2, \pi_3, \ldots, \pi_{r-1})\}$. By construction, the marginal distribution of $(\tilde i_1,\tilde j_1,\tilde\pi_2,\ldots,\tilde\pi_{r-1})$ under $\tilde D_\Pi$ is the same as the marginal distribution of $(i_1, j_1, \pi_2, \ldots, \pi_{r-1})$ under $D_{m_1', i_0', j_0'}$. Formally, for $i_1', j_1' \in [n], \pi_2', \ldots, \pi_{r-1}'\in\MS_n$,
\begin{equation}
  \label{eq:marginalequal}
\Pr_{\tilde D_\Pi} \left[ \tilde i_1  = i_1', \tilde j_1 = j_1', \tilde \pi_2 = \pi_2', \ldots, \tilde \pi_{r-1} = \pi_{r-1}'\right] = \Pr_{D_{m_1', i_0', j_0'}}\left[ i_1 = i_1', j_1 = j_1', \pi_2 = \pi_2', \ldots, \pi_{r-1} = \pi_{r-1}'\right].
\end{equation}
It is not clear a priori that the conditional distributions of $(\pi_1,\pi_r)$ under $D_{m_1', i_0', j_0'}$ and of $(\tilde \pi_1, \tilde \pi_r)$ under $\tilde D_\Pi$ are the same, since in $\tilde D_\Pi$, Alice draws $(\tilde \pi_1,\tilde \pi_r)$ with knowledge of only $(\tilde i_1,\tilde j_1,\tilde \pi_3,\tilde \pi_5,\ldots,\tilde \pi_{r-2})$, whereas under $D_{m_1',i_0',j_0'}$, $(\pi_1, \pi_r)$ is drawn from the conditional distribution with knowledge of all the permutations $(\pi_2, \pi_3, \ldots, \pi_{r-1})$. Nevertheless we will show that these two distributions are the same. More formally, for any $\pi_1',\ldots, \pi_r' \in S_n,i_1',j_1' \in [n]$,
\begin{eqnarray}
  && \Pr_{D_{m_1',i_0',j_0'}}[\pi_1 = \pi_1',  \pi_r = \pi_r' |   i_1 = i_1',  j_1 = j_1',  \pi_2 = \pi_2',  \pi_3 = \pi_3', \ldots, \pi_{r-1} = \pi_{r-1}']\nonumber\\
  \label{eq:followsfromindep}
&=& \Pr_{D_{m_1',i_0',j_0'}}[ \pi_1 = \pi_1',  \pi_r = \pi_r' |  i_1 = i_1',  j_1 = j_1',  \pi_3 = \pi_3',  \pi_5 = \pi_5', \ldots, \pi_{r-2} = \pi_{r-2}']\\
\label{eq:hatdr2}
&=& \Pr_{\tilde D_\Pi}[\tilde \pi_1 = \pi_1', \tilde \pi_r = \pi_r' | \tilde i_1 = i_1', \tilde j_1=  j_1', \tilde \pi_3 = \tilde \pi_3', \tilde \pi_5 = \pi_5', \ldots, \tilde\pi_{r-2} = \pi_{r-2}'],
\end{eqnarray}
where the second equality follows from construction (i.e., (\ref{eq:tildepi1r})), and the first equality follows from property (5) of the distribution $D \in \Dmpv(n,r,1/\log^{\beta_1}n,n/\log^{\beta_1}n)$ with $t=1$. That is, under the distribution $D$, for all $m_1', i_0', j_0', \pi_3', \ldots, \pi_{r-2}'$,
$$
(\pi_1, \pi_r) \indep (\pi_2, \pi_4, \ldots, \pi_{r-1}) | \{m_1 = m_1', i_0 = i_0', j_0 = j_0', i_1 = i_1', j_1 = j_1', \pi_3 = \pi_3', \pi_5 = \pi_5', \ldots, \pi_{r-2} = \pi_{r-2}'\}.
$$
As a consequence, under the distribution $D_{m_1', i_0', j_0'}$,
$$
(\pi_1, \pi_r) \indep (\pi_2, \pi_4, \ldots, \pi_{r-1}) | \{ i_1 = i_1', j_1 = j_1', \pi_3 = \pi_3', \pi_5 = \pi_5', \ldots, \pi_{r-2} = \pi_{r-2}' \},
$$
which verifies (\ref{eq:followsfromindep}) and therefore our claim that the distribution of $(\tilde i_1, \tilde j_1, \tilde \pi_1, \ldots, \tilde \pi_r)$ under $\tilde D_\Pi$ is the same as the distribution of $(i_1, j_1, \pi_1, \ldots, \pi_r)$ under $D_{m_1',i_0',j_0'}$.

It follows that for each tuple $(i_0', j_0', m_1')$, $\tilde \Pi$ is a protocol for the $(r-2)$-PV problem with success probability equal to: 
\begin{eqnarray}
&&  \Pr_{\tilde D_{m_1', i_0', j_0'}}[\tilde \Pi(\tilde i_1, \tilde j_1, \tilde \pi_2, \tilde \pi_3, \ldots, \tilde \pi_{r-1}) = \One[\tilde \pi_2^{r-1}(\tilde i_1) = \tilde j_1]]\nonumber\\
\label{eq:tildetsuccess}
  &= &\Pr_{D}[\Pi(i_0', j_0', \pi_1, \pi_2, \ldots, \pi_r) = \One[\pi_1^r(i_0) = j_0] | i_0 = i_0', j_0 = j_0', m_1 = m_1'].
\end{eqnarray}

\paragraph{Membership in $\Dmpv(n,r-2, 1/\log^{\beta_2}n, n/\log^{\beta_2}n)$.}
By hypothesis, we have that
$$D \in \Dmpv(n,r,1/\log^{\beta_1}n,n/\log^{\beta_1}n),$$ and that $|m_1| \leq n/\log^{\beta_1}n$. By Lemma \ref{lem:inductivestep}, for some $n_0$ that depends only on $\beta_1$ (which in turn depends only on $\beta_2$), for $n \geq n_0$, $$\tilde{D}^+ \in \Dmpvp(n,r,\log^{(\ep_2^* - \ep_1^* \beta_1)}n, n\log^{(\ep_2^* - \ep_1^* \beta_1)}n).$$
By definition of $\beta_1$, we have that $\frac{\ep_1^* \beta_1 - \ep_2^*}{2} \geq \beta_2$, so $\sqrt{\log^{(\ep_2^* - \ep_1^* \beta_1)}n} \leq 1/\log^{\beta_2}n$. We call the tuple $(i_0', j_0', m_1')$ {\bf good} if the distribution of $(i_1, j_1, \pi_2, \ldots, \pi_{r-1})$ under $\tilde{D}_{m_1', i_0', j_0'}$ belongs to $\Dmpv(n,r-2, 1/\log^{\beta_2}n, n/\log^{\beta_2}n)$. Recall that this means that

\begin{enumerate}
	\item $H(i_1 |  \pi_2, \ldots, \pi_{r-1}, m_1=m_1', i_0=i_0', j_0 = j_0') \geq \log(n) -1/\log^{\beta_2}n$. 
	\item $H(\pi_2, \ldots, \pi_{r-1} | m_1=m_1', i_0 = i_0', j_0 = j_0') \geq (r-2) \log(n!) - n/\log^{\beta_2}n.$ 
	\item $H(\One[\pi_1^r(i_1) = j_1] |  i_1, \pi_2, \ldots, \pi_{r-1}, m_1=m_1', i_0 = i_0', j_0 = j_0') \geq 1 -1/\log^{\beta_2}n$.
	\item $H(j_1 | i_1, \pi_2, \ldots, \pi_{r-1}, \pi_1^r(i_0) \neq j_0, m_1=m_1', i_0 = i_0', j_0 = j_0') \geq \log(n) -1/\log^{\beta_2}n$,
\end{enumerate}
and analogously the (b) statements in the definition of $\Dmpv(n,r,1/\log^{\beta_2}n,n/\log^{\beta_2}n)$ (\Cref{def:dmixdef}) hold as well.

By Lemma \ref{lem:inductivestep}, Markov's inequality, and a union bound, if $n \geq n_0$, with probability at least $1 - 7/\log^{\beta_2}n$ over the tuple $(i_0, j_0, m_1)$ drawn from its marginal in $D$, $(i_0, j_0, m_1)$ is good. (Notice that there is a coefficient of $7$, as opposed to $8$, since there is no (b) statement for item (2) above.)

\paragraph{Choosing a good tuple $(i_0', j_0', m_1')$.}
Now we will use (\ref{eq:tildetsuccess}) to choose a good tuple $(i_0', j_0', m_1')$ for which $\tilde \Pi$ also achieves success probability at least $1/2 + \ep_2$, for all $n > \max\left\{ n_0, 2^{(7/(\ep_1 - \ep_2))^{1/\beta_2}} \right\}$. 
For each tuple $(i_0', j_0', m_1')$, we have constructed above a protocol $\tilde \Pi$ for $(r-2)$-PV, with communication at most $n/\log^{\beta_1}n \leq n/\log^{\beta_2}n$, and where Alice and Bob use $(r-1)/2$ rounds of communication. 
If moreover $(i_0', j_0', m_1')$ is good, then the distribution of $( i_1,  j_1,  \pi_2, \ldots,  \pi_{r-1})$ under $\tilde{D}_{i_0', j_0', m_1'}$ 
belongs to $\Dmpv(n, r-2, 1/\log^{\beta_2}n,n/\log^{\beta_2}n)$.

Now suppose for the purpose of contradiction that the probability of success of all $((r-1)/2,n/\log^{\beta_2}n)$ protocols on any distribution $\tilde D \in\Dmpv(n, r-2, 1/\log^{\beta_2}n, n/\log^{\beta_2}n)$  were at most $1/2 + \ep_2$. In particular, for any good tuple $(i_0', j_0', m_1')$, the probability of success of $\tilde \Pi$ on the distribution $\tilde{D}_{m_1', i_0', j_0'}$ is at most $1/2 + \ep_2$. Then by (\ref{eq:tildetsuccess}) and since $n \geq n_0$, the probability of success of $\Pi$ would be at most
$$
7/\log^{\beta_2}n + (1 - 7/\log^{\beta_2}n) \cdot (1/2 + \ep_2) \leq 1/2 +  7 /\log^{\beta_2}n + \ep_2.
$$
Since we also have $n > 2^{(7/(\ep_1 - \ep_2))^{1/\beta_2}}$, it follows that
$$
\ep_2 + 7/\log^{\beta_2}n < \ep_1,
$$
which is a contradiction and thus completes the proof of Lemma \ref{lem:simulation}.


\end{proof}

\section*{Acknowledgements}\label{sec:ack}

We would like to thank Venkatesan Guruswami and T.S. Jayram for very enlightening discussions related to the questions considered in this work.

\bibliographystyle{alpha}
\bibliography{references,Pointer_chasing}

\end{document}